\documentclass[11pt,a4paper]{article}

\pdfoutput=1

\usepackage[a4paper, total={150mm,240mm}]{geometry}
\usepackage{amsthm}
\usepackage{amssymb}
\usepackage{amstext}
\usepackage{amsmath}
\usepackage{url}

\usepackage[margin = 40pt,font=small,labelfont=bf]{caption}
\usepackage{pdfsync}
\usepackage[T1]{fontenc}
\usepackage{color}

\usepackage{algorithm}
\usepackage{algpseudocode}

\usepackage{enumerate}
\usepackage{ae,aecompl}
\usepackage{graphicx}
\graphicspath{{./images/}}
\usepackage{caption}
\usepackage{subcaption}

\usepackage{booktabs}

\usepackage{placeins}

\usepackage{hyperref}

\numberwithin{equation}{section}

\def\argmin{\mathop{\rm arg \; min}\limits}%

\theoremstyle{plain}

\newtheorem{theorem}{Theorem}[section]
\newtheorem{lemma}{Lemma}[section]
\newtheorem{assumption}{Assumption}[section]

\theoremstyle{definition}
\newtheorem{definition}{Definition}[section]

\theoremstyle{remark}
\newtheorem{remark}{Remark}[section]

\newcommand{\R}{{\mathbb R}}

\newcommand{\Z}{{\mathbb Z}}

\newcommand{\Id}{{\mathrm{Id}}}

\font\dsrom=dsrom10 scaled 1200
\def \indic{\textrm{\dsrom{1}}}

\def\argmin{\mathop{\rm arg \; min}\limits}%
\newcommand{\supp}{\mathop{\mathrm{supp}}}

\makeatletter
\let\OldStatex\Statex
\renewcommand{\Statex}[1][3]{%
  \setlength\@tempdima{\algorithmicindent}%
  \OldStatex\hskip\dimexpr#1\@tempdima\relax}
\makeatother

\usepackage{pgfplots}
\pgfplotsset{compat=newest}
\usetikzlibrary{plotmarks}
\usepackage{grffile}

\usetikzlibrary{spy,calc}
\newif\ifblackandwhitecycle
\gdef\patternnumber{0}

\pgfkeys{/tikz/.cd,
    zoombox paths/.style={
        draw=orange,
        very thick
    },
    black and white/.is choice,
    black and white/.default=static,
    black and white/static/.style={ 
        draw=white,   
        zoombox paths/.append style={
            draw=white,
            postaction={
                draw=black,
                loosely dashed
            }
        }
    },
    black and white/static/.code={
        \gdef\patternnumber{1}
    },
    black and white/cycle/.code={
        \blackandwhitecycletrue
        \gdef\patternnumber{1}
    },
    black and white pattern/.is choice,
    black and white pattern/0/.style={},
    black and white pattern/1/.style={    
            draw=white,
            postaction={
                draw=black,
                dash pattern=on 2pt off 2pt
            }
    },
    black and white pattern/2/.style={    
            draw=white,
            postaction={
                draw=black,
                dash pattern=on 4pt off 4pt
            }
    },
    black and white pattern/3/.style={    
            draw=white,
            postaction={
                draw=black,
                dash pattern=on 4pt off 4pt on 1pt off 4pt
            }
    },
    black and white pattern/4/.style={    
            draw=white,
            postaction={
                draw=black,
                dash pattern=on 4pt off 2pt on 2 pt off 2pt on 2 pt off 2pt
            }
    },
    zoomboxarray inner gap/.initial=5pt,
    zoomboxarray columns/.initial=2,
    zoomboxarray rows/.initial=2,
    subfigurename/.initial={},
    figurename/.initial={zoombox},
    zoomboxarray/.style={
        execute at begin picture={
            \begin{scope}[
                spy using outlines={%
                    zoombox paths,
                    width=\imagewidth / \pgfkeysvalueof{/tikz/zoomboxarray columns} - (\pgfkeysvalueof{/tikz/zoomboxarray columns} - 1) / \pgfkeysvalueof{/tikz/zoomboxarray columns} * \pgfkeysvalueof{/tikz/zoomboxarray inner gap} -\pgflinewidth,
                    height=\imageheight / \pgfkeysvalueof{/tikz/zoomboxarray rows} - (\pgfkeysvalueof{/tikz/zoomboxarray rows} - 1) / \pgfkeysvalueof{/tikz/zoomboxarray rows} * \pgfkeysvalueof{/tikz/zoomboxarray inner gap}-\pgflinewidth,
                    magnification=3,
                    every spy on node/.style={
                        zoombox paths
                    },
                    every spy in node/.style={
                        zoombox paths
                    }
                }
            ]
        },
        execute at end picture={
            \end{scope}
     \gdef\patternnumber{0}
        },
        spymargin/.initial=0.5em,
        zoomboxes xshift/.initial=1,
        zoomboxes right/.code=\pgfkeys{/tikz/zoomboxes xshift=1},
        zoomboxes left/.code=\pgfkeys{/tikz/zoomboxes xshift=-1},
        zoomboxes yshift/.initial=0,
        zoomboxes above/.code={
            \pgfkeys{/tikz/zoomboxes yshift=1},
            \pgfkeys{/tikz/zoomboxes xshift=0}
        },
        zoomboxes below/.code={
            \pgfkeys{/tikz/zoomboxes yshift=-1},
            \pgfkeys{/tikz/zoomboxes xshift=0}
        },
        caption margin/.initial=4ex,
    },
    adjust caption spacing/.code={},
    image container/.style={
        inner sep=0pt,
        at=(image.north),
        anchor=north,
        adjust caption spacing
    },
    zoomboxes container/.style={
        inner sep=0pt,
        at=(image.north),
        anchor=north,
        name=zoomboxes container,
        xshift=\pgfkeysvalueof{/tikz/zoomboxes xshift}*(\imagewidth+\pgfkeysvalueof{/tikz/spymargin}),
        yshift=\pgfkeysvalueof{/tikz/zoomboxes yshift}*(\imageheight+\pgfkeysvalueof{/tikz/spymargin}+\pgfkeysvalueof{/tikz/caption margin}),
        adjust caption spacing
    },
    calculate dimensions/.code={
        \pgfpointdiff{\pgfpointanchor{image}{south west} }{\pgfpointanchor{image}{north east} }
        \pgfgetlastxy{\imagewidth}{\imageheight}
        \global\let\imagewidth=\imagewidth
        \global\let\imageheight=\imageheight
        \gdef\columncount{1}
        \gdef\rowcount{1}
        
    },
    image node/.style={
        inner sep=0pt,
        name=image,
        anchor=south west,
        append after command={
            [calculate dimensions]
            node [image container,subfigurename=\pgfkeysvalueof{/tikz/figurename}-image] {\phantomimage}
            node [zoomboxes container,subfigurename=\pgfkeysvalueof{/tikz/figurename}-zoom] {\phantomimage}
        }
    },
    color code/.style={
        zoombox paths/.append style={draw=#1}
    },
    connect zoomboxes/.style={
    spy connection path={\draw[draw=none,zoombox paths] (tikzspyonnode) -- (tikzspyinnode);}
    },
    help grid code/.code={
        \begin{scope}[
                x={(image.south east)},
                y={(image.north west)},
                font=\footnotesize,
                help lines,
                overlay
            ]
            \foreach \x in {0,1,...,9} { 
                \draw(\x/10,0) -- (\x/10,1);
                \node [anchor=north] at (\x/10,0) {0.\x};
            }
            \foreach \y in {0,1,...,9} {
                \draw(0,\y/10) -- (1,\y/10);                        \node [anchor=east] at (0,\y/10) {0.\y};
            }
        \end{scope}    
    },
    help grid/.style={
        append after command={
            [help grid code]
        }
    },
}

\newcommand\phantomimage{%
    \phantom{%
        \rule{\imagewidth}{\imageheight}%
    }%
}
\newcommand\zoombox[2][]{
    \begin{scope}[zoombox paths]
        \pgfmathsetmacro\xpos{
            (\columncount-1)*(\imagewidth / \pgfkeysvalueof{/tikz/zoomboxarray columns} + \pgfkeysvalueof{/tikz/zoomboxarray inner gap} / \pgfkeysvalueof{/tikz/zoomboxarray columns} ) + \pgflinewidth
        }
        \pgfmathsetmacro\ypos{
            (\rowcount-1)*( \imageheight / \pgfkeysvalueof{/tikz/zoomboxarray rows} + \pgfkeysvalueof{/tikz/zoomboxarray inner gap} / \pgfkeysvalueof{/tikz/zoomboxarray rows} ) + 0.5*\pgflinewidth
        }
        \edef\dospy{\noexpand\spy [
            #1,
            zoombox paths/.append style={
                black and white pattern=\patternnumber
            },
            every spy on node/.append style={#1},
            x=\imagewidth,
            y=\imageheight
        ] on (#2) in node [anchor=north west] at ($(zoomboxes container.north west)+(\xpos pt,-\ypos pt)$);}
        \dospy
        \pgfmathtruncatemacro\pgfmathresult{ifthenelse(\columncount==\pgfkeysvalueof{/tikz/zoomboxarray columns},\rowcount+1,\rowcount)}
        \global\let\rowcount=\pgfmathresult
        \pgfmathtruncatemacro\pgfmathresult{ifthenelse(\columncount==\pgfkeysvalueof{/tikz/zoomboxarray columns},1,\columncount+1)}
        \global\let\columncount=\pgfmathresult
        \ifblackandwhitecycle
            \pgfmathtruncatemacro{\newpatternnumber}{\patternnumber+1}
            \global\edef\patternnumber{\newpatternnumber}
        \fi
    \end{scope}
}

\title{Fast Wavelet Decomposition of Linear Operators through Product-Convolution Expansions}

\author{ Paul Escande \footnote{Aix Marseille Univ, CNRS, Centrale Marseille, I2M, Marseille, France, {\tt paul.escande@univ-amu.fr}} \and Pierre Weiss \footnote{Institut de Math\'{e}matiques de Toulouse, IMT-UMR5219, France and Institut des Technologies Avanc\'{e}es en Sciences du Vivant, ITAV-USR3505, CNRS and Universit\'{e} de Toulouse, Toulouse {\tt pierre.armand.weiss@gmail.com}}  }

\date{\today}

\begin{document}

\maketitle

\begin{abstract}
Wavelet decompositions of integral operators have proven their efficiency in reducing computing times for many problems, ranging from the simulation of waves or fluids to the resolution of inverse problems in imaging. 
Unfortunately, computing the decomposition is itself a hard problem which is oftentimes out of reach for large scale problems. The objective of this work is to design fast decomposition algorithms based on another representation called product-convolution expansion. This decomposition can be evaluated efficiently assuming that a few impulse responses of the operator are available, but it is usually less efficient than the wavelet decomposition when incorporated in iterative methods. The proposed decomposition algorithms, run in quasi-linear time and we provide some numerical experiments to assess its performance for an imaging problem involving space varying blurs. 
\end{abstract}

\section{Introduction}


The efficient computation of linear operators and their inverses is paramount for nearly any scientific computing problem. 
A large number of numerical approaches have been developed over the years, such as low rank decompositions, product-convolution expansions \cite{escande2016approximation}, hierarchical matrices \cite{hackbusch2013multi} or wavelet decompositions \cite{meyer1997wavelets,beylkin1991fast}. 
They can yield huge speed-ups for the practical resolution of problems ranging from partial differential equations \cite{candes2003curvelets,schneider2010wavelet} to inverse problems \cite{escande2015sparse}. We can also expect these ideas to play a growing role in the design of efficient neural networks \cite{fan2019bcr}. 

A serious hindrance to their popularization is however the high cost to perform the decomposition in a large scale setting. For a square matrix in $\R^{N\times N}$, the cost of a naive decomposition is typically $O(N^3)$, which is incompatible with many current numerical challenges which frequently satisfy $N\gg 10^6$.
The main objective of this paper is to reduce the computational burden of decomposing an operator in an orthogonal wavelet basis using an alternative decomposition called product-convolution. 

\paragraph{Product-convolution expansions}

Let $\mathcal{E}$ denote a finite dimensional vector space of discrete $d$ dimensional functions. The functions in $\mathcal{E}$ are defined on $\Omega = \{ 1, \ldots, N_1 \} \times \ldots \times \{ 1, \ldots, N_d \}$ with $N=\prod_{i=1}^d N_i$. Throughout the paper we will work with periodic boundary conditions, i.e. consider that $\Omega$ is a discretization of the torus $\mathbb{T}^{d}$.
Consider a linear operator $G$ defined for all $f\in \mathcal{E}$ by
\begin{equation} \label{eq:int_op}
	G:f \mapsto \sum_{y\in \Omega} K(\cdot,y) f(y), 
\end{equation}
where $K\in \R^{N\times N}$ is the matrix form (the kernel) of the operator. 
If all the impulse responses $S(\cdot,y) = K(\cdot+y,y)$ are well approximated by their projections on a low-dimensional subspace $\mathcal{U}=\mathrm{span}\left(u_k, 1\leq k\leq m\right)$, it is possible to construct a product-convolution expansion $H$ of $G$ defined as
\begin{equation} \label{eq:pc}
  H f = \sum_{k=1}^m u_k \star (v_k \odot f),
\end{equation}
where $\star$ denotes the convolution operator, $u_k \in \mathcal{E}$ is a convolution filter, $\odot$ denotes the point-wise multiplication and $v_k\in \mathcal{E}$ is a multiplier. The multipliers $\mathcal{V}=(v_1,\hdots, v_m)$ can be defined as the projections of the impulse responses on $\mathcal{U}$, i.e. $v_k(y)=\langle S(\cdot,y), u_k\rangle$ for all $y\in \Omega$, see \cite{escande2016approximation}. From a numerical perspective, the expansion \eqref{eq:pc} can be evaluated efficiently using fast Fourier transforms with a complexity $O(mN\log(N))$.

Product-convolution expansions have appeared at least 3 decades ago and have been given different names in various fields, see e.g. \cite{busby1981product,nagy1998restoring,flicker2005anisoplanatic,gilad2006fast,denis2015fast,alger2019scalable}.
A remarkable aspect of these expansions is their ability to interpolate linear operators from scattered impulse responses: the bases $\mathcal{U}$ and $\mathcal{V}$ can be evaluated efficiently for operators with slowly varying impulse responses. Assuming that a few impulse responses $S(\cdot,y_i)$ are known at scattered locations $y_i\in \Omega$, it is possible to evaluate the basis $\mathcal{U}$ using a principal component analysis and to interpolate the projection coefficients over the domain $\Omega$ to obtain the multipliers $\mathcal{V}$. This process not only provides a minimax estimate of the operator $G$ \cite{bigot2017estimation}, but also a compact representation compatible with fast numerical algorithms.
We refer the reader to \cite{escande2016approximation} and Section \ref{sec:pc} of this paper for more insight on the construction and approximation properties of this decomposition. 

\paragraph{Wavelet decompositions}

On the other hand, multi-scale representations of integral operators have led to practical breakthroughs in the theoretical and numerical analysis of partial differential equations and inverse problems. Despite their considerable impact for the compact representation of functions, especially in the field of signal processing, the role of wavelets for coding operators still seems marginal, at least at an industrial level. 

A possible explanation is the high cost related to the computation of the wavelet representation of an operator. Let $(\psi_{\lambda})_{\lambda \in \Lambda}$ denote an orthogonal wavelet basis and $\Psi^*: \mathcal{E}\to \R^N$ denote the associated forward wavelet transform.  
The wavelet representation of $G$ in the wavelet basis is the matrix $\Theta = \left(\theta[\lambda,\mu]\right)_{\lambda, \mu \in \Lambda}$ of coefficients $\theta[\lambda,\mu] = \langle G \psi_\lambda, \psi_\mu \rangle$. For an arbitrary operator $G$, computing these coefficients has a complexity of order $O(N^3)$ operations since the operator $G$ needs to be applied to each of the $N$ wavelets. This cost is prohibitive for large $N$. 

\paragraph{The contribution}

We address this issue by providing a fast decomposition algorithm for product-convolution operators. 
Its complexity roughly scales as $O(m N \log_2^2 N \eta^{-\alpha})$ for any arbitrary precision $\eta > 0$, where $\alpha > 0$ is a quantity that depends on the smoothness of the operator and on the number of vanishing moments of the wavelet basis. 

One may wonder what is the interest of using a wavelet decomposition if a product-convolution is already available. 
The reason comes from the higher numerical efficiency of wavelet decompositions when incorporated in iterative methods. 
For instance, we showed in \cite{escande2015sparse,escande2018accelerating} that the simultaneous sparsity of operators and images in the same orthogonal wavelet basis could be leveraged to accelerate the resolution of inverse problems by one or two orders of magnitude. The gain comes from two complementary facts:
\begin{itemize}
 \item the method can operate in the wavelet domain only and avoid to continuously swap between different domains such as the spatial domain, the Fourier domain and the wavelet domain. 
 \item In addition, efficient diagonal preconditioners can be designed since the matrix $\Theta$ is concentrated mostly along its diagonal. This idea is often referred to as a multi-level preconditioner \cite{cohen2003numerical}.
\end{itemize}
We will illustrate the power of these ideas on a practical image deblurring problem.

\section{Preliminaries} \label{sec:preliminaries}


Let $\mathcal{E}$ denote the linear space of $d$-dimensional discretized functions of $N$ samples defined on $\Omega = \{ 1, \ldots, N_1 \} \times \ldots \times \{ 1, \ldots, N_d \}$ with $\prod_{i=1}^d N_i = N$.
This space $\mathcal{E}$ will be often identified with $\R^N$ through a bijective mapping $\omega : \Omega \to \{ 1, \ldots, N \}$.
To simplify the discussion, we assume that the numbers $N_1=\hdots =N_d=2^J$ and we use circular boundary conditions. 
Assuming that the number of pixels is large enough, methods and analyses proposed in this work can be extended to any boundary conditions by using boundary wavelets.

We let  $f(x)$ denote the value of a function $f$ at $x$, $v[i]$ denote the $i$-th coefficient of a vector $v$ and $A[i,j]$ denote the $(i,j)$-th element of a matrix $A$.

Let $A$ be an $N \times N$ matrix, we let $\supp A $ denote its support i.e. the set of indexes associated to a non-zero coefficient: $ \supp A = \left\{ (i,j) \in \{1 \ldots N\}^2 \, | \, A[i,j] \neq 0 \right\}$.
The approximation rates stated in this paper will be expressed with respect to the spectral norm $\| \cdot \|_{2\to2}$ defined by
\begin{equation*}
	\| A \|_{2\to 2} = \sup_{f \in \R^N, \| f \|_2 = 1} \| A f \|_{2}.
\end{equation*}

\subsection{One dimensional orthogonal wavelet bases}

We first recall the construction of wavelets on the continuous interval $[0,1]$. 
Let $\phi$ and $\psi$ denote the scaling and mother wavelets.
Translated and dilated versions of the wavelets are defined, for $j \geq 0$, as follows
\begin{equation*} 
	\begin{split}
  \phi_{j,n} & = 2^{j/2} \phi\left( 2^{j} \cdot - n \right), \\
  \psi_{j,n} & = 2^{j/2} \psi\left( 2^{j} \cdot - n \right),
  \end{split}
\end{equation*}
with $n \in \{0,\ldots,2^j-1\}$. A mother wavelet $\psi$ is said to have $M$ vanishing moments when
\begin{equation*}
 	\forall 0 \leq m < M, \quad \int_{[0,1]} t^m \psi(t) dt = 0.
\end{equation*}

We now detail the construction of a 1D discrete orthogonal wavelet transform. A function $f : [0,1] \to \R$ can be discretized leading to a vector $v \in \R^{2^J}$  samples. The Fast Wavelet Transform is derived from the observation that wavelets $\phi$ and $\psi$ are associated to a filter bank $(h,g)$ \cite[Theorem 7.7, p. 348]{Mallat-Book}. From these filters $(h,g)$ the discrete wavelets $\phi_j$ and $\psi_j$ can be defined recursively using convolutions and sub-sampling.  
Setting $\phi_{J,n} = \indic_{\{ n \}}$, we get for $0 \leq j < J$ 
\begin{equation*}
	\phi_{j,l} = \sum_{n \in \Z} h[n - 2l] \phi_{j+1,n}, \quad \textrm{and} \quad \psi_{j,l} = \sum_{n \in \Z} g[n - 2l] \phi_{j+1,n}.
\end{equation*}

We assume that the wavelets are compactly supported i.e. $\supp(\psi)=[-\delta+1,\delta]$. 
Note that $\delta$ is related to the number of vanishing moments of the mother wavelet. Let $M$ be this number, we have $\delta \geq M$, with equality for Daubechies wavelets, see e.g. \cite[Theorem 7.9, p. 294]{Mallat-Book}. Together with \cite[Theorem 7.5, p. 286]{Mallat-Book} and $g[i] = (-1)^{1-i} h[1-i]$, we deduce that $\supp h = \supp \phi = [0,2\delta-1]$ and $\supp g = [-2(\delta-1),1]$. Therefore for $0 \leq j \leq J$,
\begin{equation*}
	\begin{split}
		\supp \phi_{j,0} &= \left[0, (2^{J-j}-1)(2\delta-1)\right] \quad \textrm{and} \\
		\supp \psi_{j,0} &= \left[-(2^{J-j}-1)2(\delta-1),(2^{J-j}-1)\right].
	\end{split}
\end{equation*}

\subsection{Orthogonal wavelet bases on \texorpdfstring{$\mathcal{E}$}{the space of signals}}

In dimension $d$, we use isotropic separable wavelet bases, see, e.g., \cite[Theorem 7.26, p. 348]{Mallat-Book}. 
Let $l=(l_1,\dots,l_d)$.
Define $\rho_{j,n}^0=\phi_{j,n}$ and $\rho_{j,n}^1=\psi_{j,n}$. 
Let $e=(e_1,\dots,e_d)\in \{0,1\}^d$ and $\mathcal{T}_j = \{ 0, \dots, 2^j - 1\}^d$.
For the ease of reading, we will use the shorthand notation $\lambda = (j,e,l)$ and $|\lambda|=j$. 
We also let $J=\frac{1}{d} \log_2 N$ and
\begin{equation}
\Lambda = \left\{ (j,e,l) \; | \; 0 \leq j \leq J-1, \; l \in \mathcal{T}_j, \; e \in \{0,1\}^d \right\}.
\end{equation}
 The wavelet $\psi_\lambda $ is defined by $\psi_{\lambda}(x_1, \ldots, x_d) = \psi_{j,l}^e(x_1,\hdots,x_d)=\rho_{j,l_1}^{e_1}(x_1)\hdots \rho_{j,l_d}^{e_d}(x_d)$.

With these definitions, every signal $f\in \mathcal{E}$ can be written as
\begin{align*}
 u & = \langle u, \psi^0_{0,0} \rangle \psi^0_{0,0} + \sum_{e\in \{0,1\}^d \setminus \{0\}} \sum_{j=0}^{J-1} \sum_{l \in \mathcal{T}_j} \langle u, \psi^e_{j,l} \rangle  \psi^e_{j,l} \\
	& = \sum_{ \lambda \in \Lambda} \langle u, \psi_{\lambda} \rangle \psi_\lambda.
\end{align*}

Finally, we let $\Psi^*: \mathcal{E} \to \R^N$ denote the discrete forward wavelet transform and $\Psi$ its inverse.
We refer to \cite{Mallat-Book, daubechies_ten_1992,cohen1993wavelets} for more details on the construction of wavelet bases.

\section{Main results} \label{sec:algo}


\subsection{Assumptions}

The main result will be stated under mild regularity and decay assumptions on the kernels $u_k$ defined below.

\begin{definition}[{Smoothness class $\mathcal{A}_\alpha$ \cite[p. 281]{cohen2003numerical}}] \label{def:smoothness}
  A convolution kernel $u_k$ is said to be $\alpha$-asymptotically smooth if there exists constants $C\geq 0$, $\alpha>0$ and a compactly supported wavelet basis with $\lceil \alpha \rceil $ vanishing moments such that the following inequality holds:
	\begin{equation}
		|\langle u_k\star \psi_\lambda,\psi_\mu \rangle| \leq C 2^{-(d/2+\alpha)\left| |\lambda| - |\mu|\right|} \vartheta(\lambda,\mu)^{-(d+\alpha)},
	\end{equation}
 	where 
	\begin{equation*}
 		\vartheta(\lambda,\mu) := 1 + 2^{\min(|\lambda|,|\mu|)} \textrm{dist}(\supp \psi_\lambda, \supp \psi_\mu)
 	\end{equation*}
 	measures the distance between the wavelets supports.
\end{definition}

\begin{remark}[Asymptotic smoothness in a continuous setting \label{rem:regularity_Winfty}]
    In a continuous setting, it can be shown that sufficient conditions for the kernels $u_k$ to be $\alpha$-asymptotically smooth is $u_k\in W^{\alpha,\infty}(\R^d)$ with the following decay property
    \begin{equation}
    |\partial^p u_k(x)| \leq C_k \left(1+\|x\|_2\right)^{-(|p|+d)},
    \end{equation}
    for all multi-indexes $p$ with $|p|\leq \alpha$. 
    Hence, the $\alpha$ smoothness property mostly boils down to regularity properties of the kernel. The larger $\alpha$, the more regular the kernels.
\end{remark}

Our main results will be stated under the following assumptions.
\begin{assumption}[Regularity of the kernel] \label{ass:all}
\
	\begin{itemize}
		\item the operator $H$ has the form \eqref{eq:pc}.
		\item there exists $\alpha > 0$ such that $u_k \in \mathcal{A_\alpha}$ for all $1 \leq k \leq m$.
		\item the mother wavelet is compactly supported on a hypercube of sidelength $2\delta$.
	\end{itemize}
\end{assumption}

Overall the set of assumptions in Assumption \ref{ass:all} describes a fairly large variety of operators comprising blurring operators \cite{escande2015sparse}, singular integral operators and pseudo-differential operators \cite{cohen2003numerical}. 

\subsection{The algorithm and its guarantees}

The main objective of this paper is to efficiently compute a sparse approximation $\widetilde{\Theta}$ of the wavelet representation $\Theta$ of $H$ defined by 
\begin{equation}
\Theta=\Psi^* H \Psi. 
\end{equation}

The proposed algorithm heavily relies on the peculiar structure of product-convolution expansions. In what follows, we let $U_k$ denote the convolution operator with $u_k$, $V_k=\mathrm{diag}(v_k)$ denote the $k$-th multiplier, $A_k=\Psi^* U_k \Psi$ denote the wavelet representation of $U_k$ and $B_k=\Psi^* V_k \Psi$ denote the wavelet representation of the multiplier. The proposed methodology is described in Algorithm \eqref{alg:decomp}. 
\begin{algorithm}
	\begin{algorithmic}
	\Require A precision $\eta > 0$, the 
	\State Set $\epsilon_k = \frac{\eta}{m \| v_k \|_{\infty}} $ for all $1 \leq k \leq m$
	\State Set $\widetilde{\Theta}$ an empty sparse matrix
	\ForAll{$k = 1 \to m$}
		\State Compute $A_k$ (by exploiting the convolution structure) \Comment{$O(N \log^2 N)$}
		\State Threshold $A_k$ to get an $\epsilon_k$ approximation $\widetilde{A}_{k}$ \Comment{$O(N \log N)$}
		\State Construct $B_k$ (using a sparse cascade algorithm) \Comment{$O(N \log N)$}
		\State Compute $C_k = \widetilde{A}_k B_k$ \Comment{$O(N \log N \epsilon_k^{-d/t})$}
		\State Accumulate $\widetilde{\Theta} \gets \widetilde{\Theta} + C_k $ \Comment{$O(N \log N \epsilon_k^{-d/t})$}
	\EndFor
	\State \Return $\widetilde\Theta$ an $\eta$-approximation of $\Theta$
	\end{algorithmic}
\caption{Decomposition of product-convolution in an orthogonal wavelet basis} \label{alg:decomp}
\end{algorithm}
\FloatBarrier

Our main theoretical result reads as follows.
\begin{theorem}\label{thm:overall_complexity}
	Under Assumption \ref{ass:all}, Algorithm \ref{alg:decomp} produces a matrix $\widetilde\Theta$ satisfying $\| \Theta - \widetilde\Theta\|_{2\to 2} \leq \eta$ in $O\left(\delta^{d+1} m^{\max(1,d/\alpha)} N \log^2_2 N \eta^{-d/\alpha} \right)$ operations.
	Furthermore, the number of coefficients in $\widetilde\Theta$ can also be bounded above by $O\left(\delta^{d} N \log^2_2 N \eta^{-d/\alpha} \right)$.
\end{theorem}

\begin{remark}
	In many applications, the kernels of integral operators are smooth. 
	In particular, for $\alpha\geq d$ the term $\max(1,d/\alpha) = 1$. The complexity is thus linear in $m$.
	Remark \ref{rem:regularity_Winfty} shows that this is satisfied when the kernels $u_k$ are of class $W^{d,\infty}$.
\end{remark}

\begin{remark}
	The bound on the number of coefficients in $\widetilde\Theta$ provided by Theorem \ref{thm:overall_complexity} allows to bound the complexity of matrix-vector products with $\widetilde\Theta$. 
	If we assume that the operator $G$ belongs to the class $\mathcal{A}_\alpha$, a direct application of  \cite{beylkin1991fast,escande2015sparse} (see also Theorem \ref{thm:cohen}), shows that one can construct an $\eta$ approximation $\widetilde G$ of $G$, containing no more than $O(\delta^d N \eta^{-d/\alpha})$ coefficients.
	The result of Theorem \ref{thm:overall_complexity} is therefore optimal up to the $\log_2 N$ factor. We believe that it could be discarded by assuming further regularity of the multipliers $v_k$. We could then use an additional approximation in place of $B_k$, in the exact same fashion as for the $A_k$'s.
\end{remark}

\subsection{Getting a product-convolution expansion} \label{sec:pc}

While the last two items in Assumption \ref{ass:all} are now well established, the combination with the first one may seem problematic. We review some numerical approaches to design product-convolution expansions below.

\subsubsection{Naive interpolation of impulse responses} \label{sec:pc_impulse}
	The simplest and probably the most widespread product-convolution expansions in image processing are based on the assumption of slow variations of the impulse responses of the operator in the domain $\Omega$. Under this assumption, it is possible to set
	\begin{equation*}
		u_k = S(\cdot,y_k)
	\end{equation*}
	i.e. $(u_k)_{k=1}^m$ are the impulse responses of the operator at some locations $(y_k)_{k=1}^m$. 
	By setting the locations $(y_k)_{k=1}^m$ as a coarse Euclidean grid \cite{nagy1998restoring}, it is possible to keep a small value for $m$. The functions $(v_k)_{k=1}^m$ are then used to interpolate the impulse responses i.e. they define a partition of unity with the constraints $v_k(y_k) = 1$ for all $1 \leq k \leq m$.
	The choice of these $v_k$ are discussed in many works, without particular guarantees on the approximation obtained. 
	We refer the interested reader to \cite{denis2015fast} for a nice overview and \cite{escande2016approximation} for approximation rates.

	In the context of PDEs, product-convolution expansions are encountered in Schur complement methods for solving PDEs, Dirichlet-to-Neumann maps and in PDE-constrained optimization as Hessians.
	In \cite{alger2019scalable}, the functions $(u_k,v_k)_{k=1}^m$ are chosen as above, but instead of being defined on a Euclidean grid, locations $(y_k)_{k=1}^m$ are adaptively sampled to best approximate the operator with a fixed number $m$.

\subsubsection{Interpolating expansions of impulse responses} \label{sec:pc_interp}
	Under mild regularity assumptions of the impulse responses, product-convolution expansions can be designed from scattered impulse responses in a more efficient fashion than what was just described \cite{flicker2005anisoplanatic,gentile2013interpolating,bigot2017estimation}.

	The main idea is to use a principal component analysis of the observed impulse responses $S(\cdot,y_k)$ to obtain a low-dimensional orthogonal basis $(u_k)$. The coefficient maps $(v_k)$ can then be computed by interpolating the known projection coefficients at the positions $(y_k)$. We showed in \cite{bigot2017estimation} that this approach was optimal from an approximation theoretic point of view.
	
\subsubsection{Singular value decomposition} \label{sec:pc_known_op}
	When the operator $G$ in \eqref{eq:int_op} is fully known, the optimal way (in Frobenius norm) for fixed $m$ to get a product-convolution expansion of an operator \eqref{eq:int_op} is to construct a low-rank approximation of the \emph{SVIR $S$} (and not the kernel $K$). 
	For many practical problems, the SVIR $S$ has a much lower rank than the kernel $K$. A typical example is the identity operator which has a kernel of rank $N$ and an SVIR of rank 1. We refer to \cite{escande2016approximation} for more examples and mathematical insights.
	
	The low-rank decomposition of space varying impulse response can be achieved by performing a singular value decomposition (SVD) on $S$.
	For most problems, this approach is however intractable due to the size of matrices $S$. 
	
	Alternatively, we may just assume that the action of an operator with a matrix form $S$ can be evaluated on any vector. The couples $(u_k,v_k)_{k=1}^m$ can now be obtained using a randomized SVD of the matrix $S$ \cite{halko2011finding}. The complexity of such algorithms scales as $O(m^2 N)$, where $m$ is the number of matrix-product evaluations. 

\section{Numerical experiments -- Spatially varying deblurring} \label{sec:numerics}


\subsection{Description of the operators}

To assess the performance of the algorithm, we propose to perform numerical experiments on a large scale image deblurring problem. We synthesize two different blurring operators:
\begin{itemize}
	\item The first one, in Figure \ref{fig:psf_blur_1}, is made of isotropic Gaussian impulse responses with a variance $\sigma(y_1,y_2)$ increasing along the vertical direction only (i.e. $\sigma(y_1,y_2) = 3 y_2$ for $(y_1,y_2) \in [0,1]^2$). The impulse responses are then truncated out, so that 99\% of the Gaussian's energy is kept. 
	This allows to encode the operator as a sparse matrix to accelerate the explicit computation of $\Theta$. Notice that this is however unnecessary to apply our algorithms. As an indication, the largest support has a size $19 \times 19$ pixels.
	\item The second operator, in Figure \ref{fig:psf_blur_2} models realistic degradations appearing in optical systems \cite{simpkins2014parameterized}. The impulse responses are rotated and skewed Gaussian, with parameters depending on their location with respect to the center of the field. As an indication, the support of the impulse responses are of size $25 \times 25$.
\end{itemize}

\begin{figure}[htpb] \centering
\begin{subfigure}[b]{0.9\textwidth} \centering
\begin{tikzpicture}[zoomboxarray]
    \node [image node] { \includegraphics[width=0.45\textwidth]{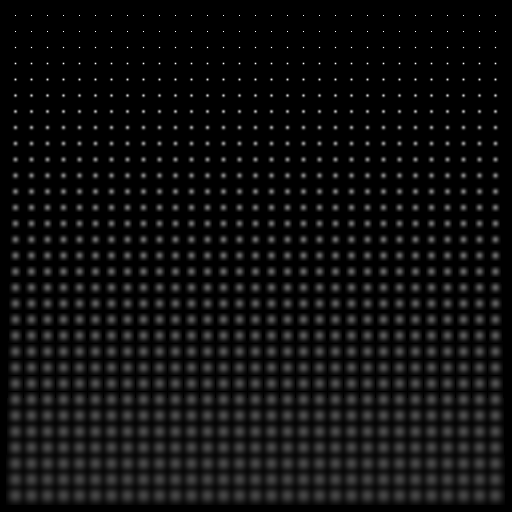} };
    \zoombox[color code=red]{0.5,0.5}
    \zoombox[color code=blue]{0.5,0.9}
    \zoombox[magnification=3]{0.1,0.1}
    \zoombox[color code=green]{0.7,0.3}
\end{tikzpicture}
\caption{Operator 1 (applied to a $1024\times 1024$ image).} \label{fig:psf_blur_1}
\end{subfigure}

\begin{subfigure}[b]{0.9\textwidth} \centering
\begin{tikzpicture}[zoomboxarray] 
    \node [image node] { \includegraphics[width=0.45\textwidth]{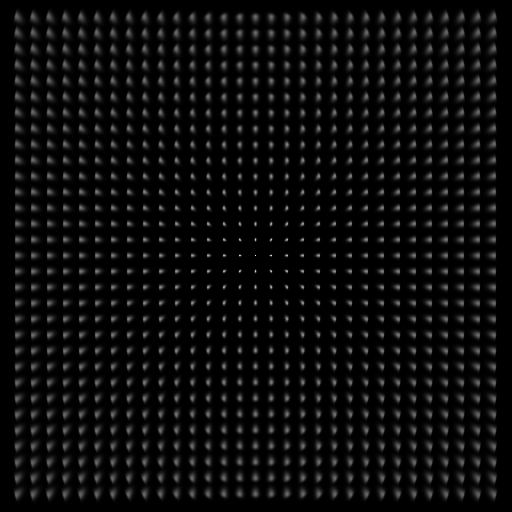} };
    \zoombox[color code=red]{0.5,0.5}
    \zoombox[color code=blue]{0.5,0.9}
    \zoombox[magnification=3]{0.1,0.1}
    \zoombox[color code=green]{0.7,0.3}
\end{tikzpicture}
\caption{Operator 2 (applied to a $1024\times 1024$ image)} \label{fig:psf_blur_2}
\end{subfigure}
  \caption{An illustration of the spatially varying blurs used in the numerical experiments.
   Blurring operators are applied to a Dirac comb to obtain the impulse responses at various locations.} \label{fig:varying_blur}
\end{figure}

We compute the product-convolution expansions of each of these operators using a randomized SVD of the matrix $S$ as explained in Section \ref{sec:pc_known_op}. The order $m$ of the expansion has been minimized so as to obtain a fixed deblurring performance, see paragraph \ref{sec:choice_m}. The first operator - Figure \ref{fig:psf_blur_1}  - is decomposed with $m=5$ coefficients, while the second one - Figure \ref{fig:psf_blur_2} - with $m=25$ coefficients.

\subsection{Comparing decomposition timings}

In this section, we compare Algorithm \ref{alg:decomp} to standard decomposition methods. 
The direct algorithm to compute $\Theta$ from a product-convolution expansion is given in Algorithm \ref{alg:decomp_naive}.
It consists in using a matrix-vector product and a wavelet transform for each column. Its complexity is therefore in $O(N^3)$ operations (this complexity is reduced here to $O(N^2)$ since we use the fact that the impulse responses are compactly supported).

\begin{algorithm}
    \begin{algorithmic}
    \ForAll{$\lambda \in \Lambda$}
        \State Get $\psi_\lambda$ and compute $w_\lambda=H \psi_\lambda$
        \State Compute $\Psi^*w_\lambda$ to obtain $(\langle H \psi_\lambda, \psi_\mu \rangle)_{\mu \in \Lambda}$. Set it as the $\lambda$-th column of $\Theta$
    \EndFor
    \end{algorithmic}
\caption{Naive decomposition of matrices in an orthogonal wavelet basis} \label{alg:decomp_naive}
\end{algorithm}
\FloatBarrier

In this comparison, we used parallel implementations running on 12 cores on a workstation with 256GB of RAM in double precision. Figure \ref{fig:timings_wrt_N} compares the computing times using 3 Algorithms:
\begin{itemize}
 \item Algorithm \ref{alg:decomp_naive} with a direct computation with a sparse matrix (Alg. \ref{alg:decomp_naive} -- exact).
 \item Algorithm \ref{alg:decomp_naive} with a product-convolution expansion (Alg. \ref{alg:decomp_naive} -- PC). 
 \item The proposed Algorithm \ref{alg:decomp} (Alg. \ref{alg:decomp}).
\end{itemize}

We use square images containing $N = 2^n \times 2^n$ pixels with $n \in \{ 7, \ldots, 12\}$. The maximal image size is therefore $4096\times 4096$. 
We run the algorithms with a Symmlet wavelet basis of order 6. This choice is explained in Remark \ref{rem:choice_wavelet}.
The precision $\eta= 5.10^{-4}$ has been chosen so as to obtain good performance for a deblurring experiment. This choice is further detailed in Section \ref{sec:choice_m}.

In this setting, Algorithm \ref{alg:decomp} is much faster than any instance of Algorithm \ref{alg:decomp_naive}. 
For $n=12$, corresponding to images of 16 millions pixels, the speed-up from Alg. \ref{alg:decomp_naive} -- product-convolution to Alg. \ref{alg:decomp} is 4680 for the operator on Figure \ref{fig:psf_blur_1} and 2330 for the operator on Figure \ref{fig:psf_blur_2}. 

We also evaluate the decomposition times for various precisions with $N = 256 \times 256$ being fixed, see Figure \ref{fig:timings_wrt_error}. 
Assuming that the clock time is proportional to the number of operations, we deduce that the number of operations is bounded by a quantity proportional to $\eta^{-d/\alpha }$. We fit the curves with a line to obtain an approximation of the smoothness $\alpha$ of each operator.

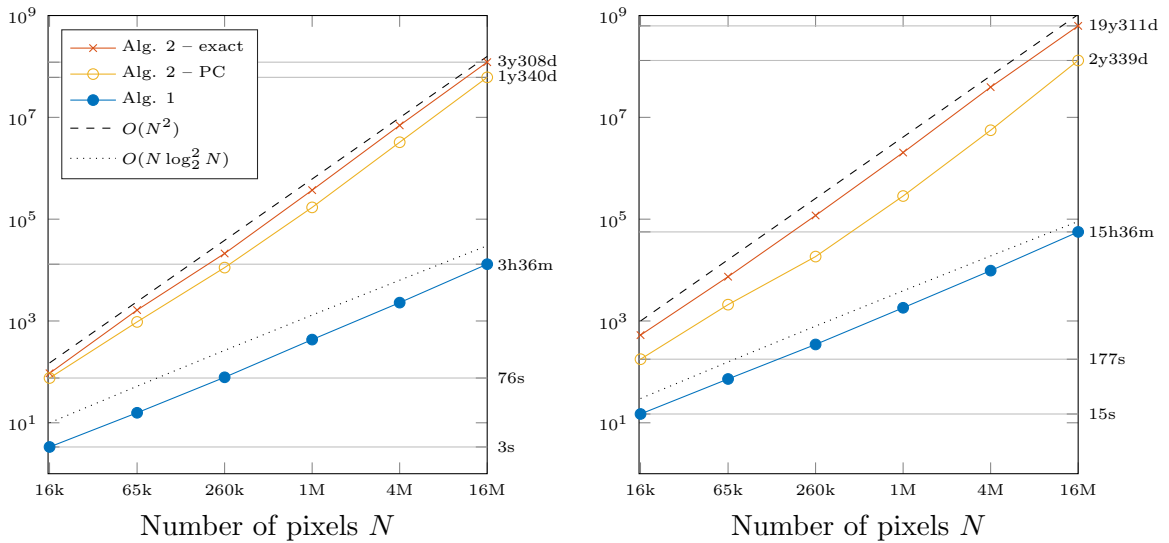
\begin{figure}[htpb] \centering
\begin{subfigure}[b]{0.45\textwidth}
%
%
\definecolor{mycolor1}{rgb}{0.00000,0.44700,0.74100}%
\definecolor{mycolor2}{rgb}{0.85000,0.32500,0.09800}%
\definecolor{mycolor3}{rgb}{0.92900,0.69400,0.12500}%
\definecolor{mycolor4}{rgb}{0.49400,0.18400,0.55600}%
\definecolor{mycolor5}{rgb}{0.46600,0.67400,0.18800}%
\begin{tikzpicture}

\begin{axis}[%
width=0.856\textwidth,
height=0.9\textwidth,
at={(0\textwidth,0\textwidth)},
scale only axis,
xmode=log,
xmin=16000,
xmax=16777216,
xminorticks=true,
ymode=log,
ymin=1,
ymax=1000000000,
yminorticks=true,
xtick={16384,65536,262144,1048576,4194304,16777216},
xticklabels={16k,65k,260k,1M,4M,16M},
ticklabel style = {font=\tiny},
xlabel={Number of pixels $N$},
extra y ticks={121278558.529126, 13048, 3.351279,75.7858304, 61064333.3873664},
extra y tick labels={3y308d, 3h36m, 3s, 76s, 1y340d},
extra y tick style={grid=major, yticklabel pos=right},
axis background/.style={fill=white},
legend style={legend cell align=left,align=left,draw=white!15!black,font=\tiny},
legend pos={north west},
]

\addplot [color=mycolor2,solid,mark=x]
  table[row sep=crcr]{%
16384	93.9606016\\
65536	1639.2454144\\
262144	21179.531264\\
1048576	371573.4962176\\
4194304	6958800.8042496\\
16777216	121278558.529126\\
};
\addlegendentry{Alg. \ref{alg:decomp_naive} -- exact};

\addplot [color=mycolor3,solid,mark=o]
  table[row sep=crcr]{%
16384	75.7858304\\
65536	958.3919104\\
262144	11185.1864064\\
1048576	170320.8288256\\
4194304	3242496.884736\\
16777216	61064333.3873664\\
};
\addlegendentry{Alg. \ref{alg:decomp_naive} -- PC};

\addplot [color=mycolor1,solid,mark=*]
  table[row sep=crcr]{%
16384	3.351279\\
65536	15.728688\\
262144	78.30821\\
1048576	430.931303\\
4194304	2292.405213\\
16777216	13048.004146\\
};
\addlegendentry{Alg. \ref{alg:decomp}};

\addplot [color=black,dashed]
  table[row sep=crcr]{%
16384	150\\
65536	2400\\
262144	38400\\
1048576	614400\\
4194304	9830400\\
16777216 157286400\\
};
\addlegendentry{$O(N^2)$};

\addplot [color=black,dotted]
  table[row sep=crcr]{%
16384	10 \\
65536	52.2448979591837\\
262144	264.489795918367\\
1048576	1306.12244897959\\
4194304	6321.63265306122\\
16777216 30093.0612244898\\
};
\addlegendentry{$O(N \log_2^2 N)$};

\end{axis}
\end{tikzpicture}%
\end{subfigure}
$\qquad$
\begin{subfigure}[b]{0.45\textwidth}
%
%
\definecolor{mycolor1}{rgb}{0.00000,0.44700,0.74100}%
\definecolor{mycolor2}{rgb}{0.85000,0.32500,0.09800}%
\definecolor{mycolor3}{rgb}{0.92900,0.69400,0.12500}%
\definecolor{mycolor4}{rgb}{0.49400,0.18400,0.55600}%
\definecolor{mycolor5}{rgb}{0.46600,0.67400,0.18800}%
\begin{tikzpicture}

\begin{axis}[%
width=0.856\textwidth,
height=0.9\textwidth,
at={(0\textwidth,0\textwidth)},
scale only axis,
xmode=log,
xmin=16000,
xmax=16777216,
xminorticks=true,
ymode=log,
ymin=1,
ymax=1000000000,
yminorticks=true,
xtick={16384,65536,262144,1048576,4194304,16777216},
xticklabels={16k,65k,260k,1M,4M,16M},
xlabel={Number of pixels $N$},
ticklabel style = {font=\tiny},
extra y ticks={626116523.613749248, 56192.864093, 14.889932, 177.544011, 130942463.115264},
extra y tick labels={19y311d, 15h36m, 15s, 177s, 2y339d},
extra y tick style={grid=major, yticklabel pos=right},
axis background/.style={fill=white},
legend style={legend cell align=left,align=left,draw=white!15!black,font=\tiny},
legend pos={north west}
]

\addplot [color=mycolor2,solid,mark=x]
  table[row sep=crcr]{%
16384	527.8646272\\
65536	7415.4377216\\
262144	118178.5522176\\
1048576	2034835.4428928\\
4194304	39205793.5888384\\
16777216	626116523.613749248\\
};
\addlegendentry{Spatial matrix};

\addplot [color=mycolor3,solid,mark=o]
  table[row sep=crcr]{%
16384	177.544011\\
65536	2093.0101248\\
262144	18414.5936384\\
1048576	284703.064064\\
4194304	5576765.472768\\
16777216	130942463.115264\\
};
\addlegendentry{PC};

\addplot [color=mycolor1,solid,mark=*]
  table[row sep=crcr]{%
16384	14.889932\\
65536	72.557557\\
262144	345.845038\\
1048576	1816.567624\\
4194304	9747.712592\\
16777216	56192.864093\\
};
\addlegendentry{Algo};

\addplot [color=black,dashed]
  table[row sep=crcr]{%
16384	1000\\
65536	16000\\
262144	256000\\
1048576	4096000\\
4194304	65536000\\
16777216 1048576000\\
};
\addlegendentry{$O(N^2)$};

\addplot [color=black,dotted]
  table[row sep=crcr]{%
16384	30\\
65536	156.734693877551\\
262144	793.469387755102\\
1048576	3918.36734693878\\
4194304	18964.8979591837\\
16777216 90279.1836734694\\
};
\addlegendentry{$O(N \log_2^2 N)$};

\legend{} ;

\end{axis}
\end{tikzpicture}%
\end{subfigure}
  \caption{Running times for various number of pixels  $N = 2^n \times 2^n$ with $n \in \{ 7, \ldots, 12\}$. Left: for the blurring operator on Figure \ref{fig:psf_blur_1} and Right: for the one on Figure \ref{fig:psf_blur_2}.} \label{fig:timings_wrt_N}
\end{figure}

\begin{figure}[htpb] \centering
\begin{subfigure}[b]{0.45\textwidth}
%
%
\definecolor{mycolor1}{rgb}{0.00000,0.44700,0.74100}%
\definecolor{mycolor2}{rgb}{0.85000,0.32500,0.09800}%
\begin{tikzpicture}

\begin{axis}[%
width=0.856\textwidth,
height=0.9\textwidth,
at={(0\textwidth,0\textwidth)},
scale only axis,
xmode=log,
xmin=1e-11,
xmax=1,
xmajorgrids,
ymajorgrids,
xminorticks=true,
ymode=log,
ymin=10,
ymax=1e5,
yminorticks=true,
xlabel={Precision $\eta$},
ylabel={Time (s)},
axis background/.style={fill=white},
legend style={legend cell align=left,align=left,draw=white!15!black,font=\tiny},
ticklabel style = {font=\tiny},
]
\addplot [color=mycolor1,solid,thick]
  table[row sep=crcr]{%
1.80784221382015e-11	2162.941625\\
2.24120703660194e-11	2201.88666\\
3.70055216859638e-11	2010.644929\\
7.36275248607216e-11	1927.169559\\
1.5445904319607e-10	1840.186374\\
3.22143260858928e-10	1748.75678\\
6.96059473236566e-10	1651.711343\\
1.45499484995441e-09	1547.723732\\
3.05141215816243e-09	1443.06721\\
6.35300771649496e-09	1327.119255\\
1.34362259091326e-08	1213.192979\\
2.82732088994837e-08	1098.710265\\
5.89040027061174e-08	982.721583\\
1.23992415430655e-07	871.699326\\
2.52252663043278e-07	761.528742\\
5.1168061939953e-07	654.597466\\
1.01015836286354e-06	646.319091\\
2.09074516630877e-06	479.446291\\
4.24367780382086e-06	382.470091\\
8.41055525408891e-06	310.600171\\
1.73531098477289e-05	248.834622\\
3.55930216754963e-05	200.487555\\
7.32998068300962e-05	155.392744\\
0.00015080789129171	121.5648\\
0.000299616042322034	95.227028\\
0.000597828534686361	75.330738\\
0.00112655466092415	58.862097\\
0.00214649449839959	47.946383\\
0.00398225677917674	38.736334\\
0.00719012773593595	32.272336\\
0.0127089325717889	26.691588\\
0.0221156832147235	22.729819\\
0.0367742714766127	19.641446\\
0.0590945642506985	16.90793\\
0.0959386808917537	15.685182\\
0.143390801836229	13.767977\\
0.227421695294071	13.262766\\
0.400282193225295	11.018766\\
0.608285049166736	9.66555\\
1.00833533644395	8.716306\\
1.01648109450941	7.273938\\
};
\addlegendentry{Blur Fig. \ref{fig:psf_blur_1}};

\addplot [color=black,dashed, thick]
  table[row sep=crcr]{%
1.67501726218833e-11	19264.8068335668\\
1.67364628533338e-11	19269.8044150279\\
1.67527812506847e-11	19263.8565277259\\
1.67762119939259e-11	19255.329592045\\
1.67456386191878e-11	19266.4590039101\\
1.67714968510424e-11	19257.0442699577\\
1.67844326561904e-11	19252.3416371009\\
1.67375293035589e-11	19269.4154717505\\
1.68909903039785e-11	19213.7847828164\\
1.80784221382015e-11	18804.6978182578\\
2.24120703660194e-11	17567.2789767165\\
3.70055216859638e-11	14987.0410280753\\
7.36275248607216e-11	12052.357468622\\
1.5445904319607e-10	9531.08770403349\\
3.22143260858928e-10	7551.21299539084\\
6.96059473236566e-10	5915.95436385849\\
1.45499484995441e-09	4683.69732215675\\
3.05141215816243e-09	3704.26154679786\\
6.35300771649496e-09	2936.40367884421\\
1.34362259091326e-08	2316.16589348959\\
2.82732088994837e-08	1829.87112243501\\
5.89040027061174e-08	1450.24833843336\\
1.23992415430655e-07	1145.63139671607\\
2.52252663043278e-07	914.826561023654\\
5.1168061939953e-07	731.201849680497\\
1.01015836286354e-06	589.474595029259\\
2.09074516630877e-06	468.157870380353\\
4.24367780382086e-06	374.112958714703\\
8.41055525408891e-06	301.227525188138\\
1.73531098477289e-05	239.470699808425\\
3.55930216754963e-05	190.731609115174\\
7.32998068300962e-05	151.718475833632\\
0.00015080789129171	120.721960638534\\
0.000299616042322034	97.1276325884361\\
0.000597828534686361	78.0381155784472\\
0.00112655466092415	63.8466537499058\\
0.00214649449839959	52.0533150575115\\
0.00398225677917674	42.798275594561\\
0.00719012773593595	35.4927456646657\\
0.0127089325717889	29.633182562862\\
0.0221156832147235	24.8636625482621\\
0.0367742714766127	21.1644718734656\\
0.0590945642506985	18.2117124897555\\
0.0959386808917537	15.6202069874729\\
0.143390801836229	13.7531094728618\\
0.227421695294071	11.8835884195435\\
0.400282193225295	9.93501579934392\\
0.608285049166736	8.70157125830534\\
1.00833533644395	7.41423508353158\\
1.01648109450941	7.39536206036588\\
};
\addlegendentry{$\eta^{-2/(3.15)}$};

\addplot [color=red,solid, thick]
  table[row sep=crcr]{%
1.96014977027114e-11	32201.96356\\
3.22151403906156e-11	31496.760059\\
6.54721516295325e-11	30210.573064\\
1.28728146958336e-10	29003.620689\\
2.6092747169026e-10		28500.190644\\
5.27697571152929e-10	26135.198951\\
1.03463779699788e-09	23133.739256\\
1.99386821901577e-09	20765.766415\\
4.02346942635925e-09	18870.926721\\
8.61095728436681e-09	16687.045337\\
1.80745959003242e-08	14317.79797\\
3.76358141870151e-08	12252.925161\\
7.8204137991114e-08		10752.688076\\
1.60615836592298e-07	9342.641878\\
3.24574432356625e-07	8519.849943\\
6.28248047584408e-07	8143.282118\\
1.26552623071018e-06	8065.144066\\
2.55392035619669e-06	6272.541874\\
4.97813979745638e-06	5689.8585\\
9.79367551635203e-06	4460.866195\\
1.94864627124844e-05	3452.408513\\
3.79528350106531e-05	2630.444424\\
7.31788553470079e-05	1968.1884\\
0.000140940145669814	1442.310285\\
0.000275792218755697	1051.769858\\
0.000531076043196739	754.116093\\
0.000998095494895607	543.171327\\
0.00187872145546865	398.47261\\
0.00350691012309381	290.929708\\
0.00622745406796776	216.499917\\
0.0110461100914335	161.129597\\
0.0197370939176132	125.092024\\
0.0323013391131282	103.266805\\
0.0546318523832499	90.950308\\
0.0913161434184625	80.310931\\
0.142431501633974	69.223201\\
0.208297248246717	60.743267\\
0.320945060282772	51.836199\\
0.511622405906813	50.461508\\
0.765907116524617	44.405385\\
1.2241919575351	36.253761\\
1.24318878339538	32.610694\\
};
\addlegendentry{Blur Fig. \ref{fig:psf_blur_2}};

\addplot [color=black,dotted, thick]
  table[row sep=crcr]{%
1.71312625864068e-11	1006720.35221768\\
1.73677362502652e-11	1000961.90379117\\
1.73750875488879e-11	1000784.67402971\\
1.73090861382046e-11	1002379.69900372\\
1.73876805692326e-11	1000481.32054595\\
1.74075245333766e-11	1000003.93086494\\
1.75136167231898e-11	997464.685585773\\
1.79466968124948e-11	987321.234669063\\
1.96014977027114e-11	951547.375776146\\
3.22151403906156e-11	772937.26927483\\
6.54721516295325e-11	574470.219711812\\
1.28728146958336e-10	432920.111444691\\
2.6092747169026e-10	322116.045010024\\
5.27697571152929e-10	239898.448738696\\
1.03463779699788e-09	180998.726673017\\
1.99386821901577e-09	137549.757797957\\
4.02346942635925e-09	102536.158504068\\
8.61095728436681e-09	74576.9643400917\\
1.80745959003242e-08	54684.1246943635\\
3.76358141870151e-08	40232.3909756002\\
7.8204137991114e-08	29625.7131301953\\
1.60615836592298e-07	21922.1508809304\\
3.24574432356625e-07	16332.0326299582\\
6.28248047584408e-07	12388.7037755101\\
1.26552623071018e-06	9241.93417813455\\
2.55392035619669e-06	6889.16707435326\\
4.97813979745638e-06	5210.49628464708\\
9.79367551635203e-06	3925.63412178515\\
1.94864627124844e-05	2943.65279827513\\
3.79528350106531e-05	2227.12602741934\\
7.31788553470079e-05	1692.12129802106\\
0.000140940145669814	1286.24642295776\\
0.000275792218755697	971.248633201613\\
0.000531076043196739	738.336334689117\\
0.000998095494895607	567.017491129622\\
0.00187872145546865	435.167308731975\\
0.00350691012309381	335.146105764928\\
0.00622745406796776	263.561667089989\\
0.0110461100914335	207.364361490373\\
0.0197370939176132	162.651269086379\\
0.0323013391131282	132.354426544371\\
0.0546318523832499	106.228449939375\\
0.0913161434184625	85.6815241064554\\
0.142431501633974	71.138489923013\\
0.208297248246717	60.6778734394055\\
0.320945060282772	50.6372671853551\\
0.511622405906813	41.6609494887964\\
0.765907116524617	35.1890985585249\\
1.2241919575351	28.9190443968033\\
1.24318878339538	28.7333077577642\\
};
\addlegendentry{$\eta^{-2/(2.38)}$};

\end{axis}
\end{tikzpicture}%
\end{subfigure}
\caption{Decomposition times for various precisions and $N = 256 \times 256$.} \label{fig:timings_wrt_error}
\end{figure}
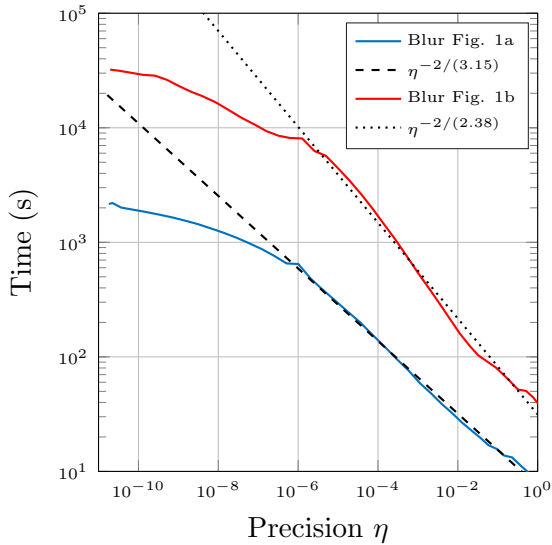

\begin{remark}\label{rem:choice_wavelet}
    The choice of the wavelet basis is of crucial importance as it has to simultaneously provide a good representation for the image and for the operator.
    A key parameter is the number of vanishing moments of the wavelet basis. As suggested by Definition \ref{def:smoothness}, the number of vanishing moments should be at least equal to the smoothness $\alpha$ of the convolution kernels $u_k$ in order to get the best possible approximation.
    
    In \cite{escande2015sparse} we performed  extensive numerical experiments of the approximation properties of a wavelet basis w.r.t. the number of vanishing moments. As expected, taking as many moments as possible was preferable for the approximation rate. However, increasing it too much led to insignificant  approximation quality while deteriorating the numerical complexity significantly. This is mostly due to the constants of Theorem \ref{thm:cohen} which increase with the number of vanishing moments.

    We performed additional numerical simulations - not reported here - to assess the behavior of the complexity of Algorithm \ref{alg:decomp} w.r.t. the sidelength $\delta$ of the wavelet basis. This sidelength is related to the number of vanishing moments $M$ since $\delta \geq M$ \cite[Theorem 5.7, p. 286]{Mallat-Book}. Numerically, we observed that the algorithm actually performs better than the rate $\delta^{d+1}$ in Theorem \ref{thm:overall_complexity}.
\end{remark}

\subsection{A deblurring experiment} \label{sec:deblurring}

\paragraph{The setting}

We assume that an observed image $f_0\in \mathcal{E}$ reads
\begin{equation*}
    f_0 = Hf + b,
\end{equation*}
where $f \in \mathcal{E}$ is the clean image to recover, $b \sim \mathcal{N}(0, \sigma^2 \textrm{Id}_N)$ is a white Gaussian noise of standard deviation $\sigma$ and $H \in \R^{N \times N}$ is the blurring operator. The image $f$ used in this experiment is shown in Figure \ref{fig:confocal} and Figure \ref{fig:images_original_degraded} contains the degraded version $f_0$ with the two operators considered.

\begin{figure}[htpb] \centering
\begin{subfigure}[b]{0.9\textwidth} \centering
\begin{tikzpicture}[zoomboxarray] 
    \node [image node] { \includegraphics[width=0.4\textwidth]{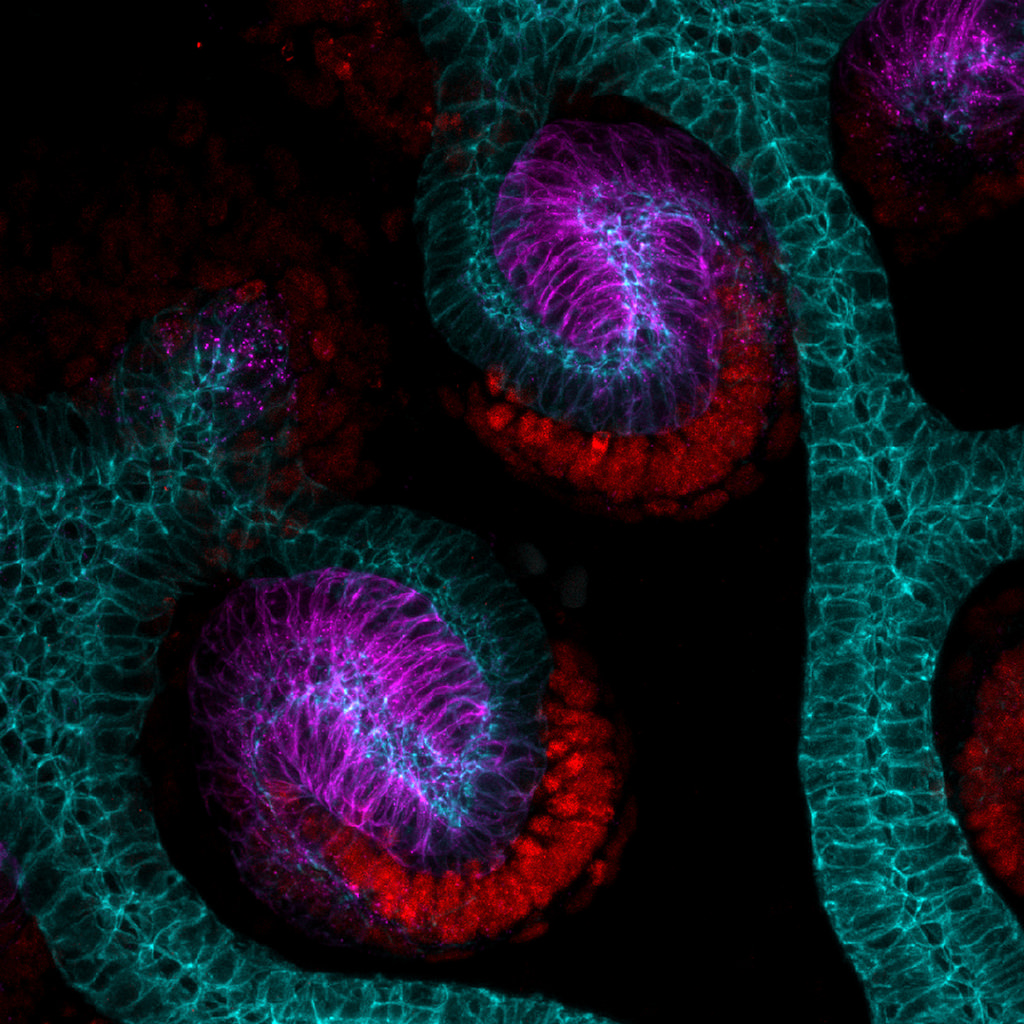} };
    \zoombox[color code=red]{0.4,0.4}
    \zoombox[color code=blue]{0.5,0.9}
    \zoombox[magnification=3]{0.1,0.1}
    \zoombox[color code=green]{0.9,0.5}
\end{tikzpicture}
\caption{Original image $f$ -- $1024 \times 1024$ pixels} \label{fig:confocal}
\end{subfigure}

\begin{subfigure}[b]{0.9\textwidth} \centering
\begin{tikzpicture}[zoomboxarray] 
    \node [image node] { \includegraphics[width=0.4\textwidth]{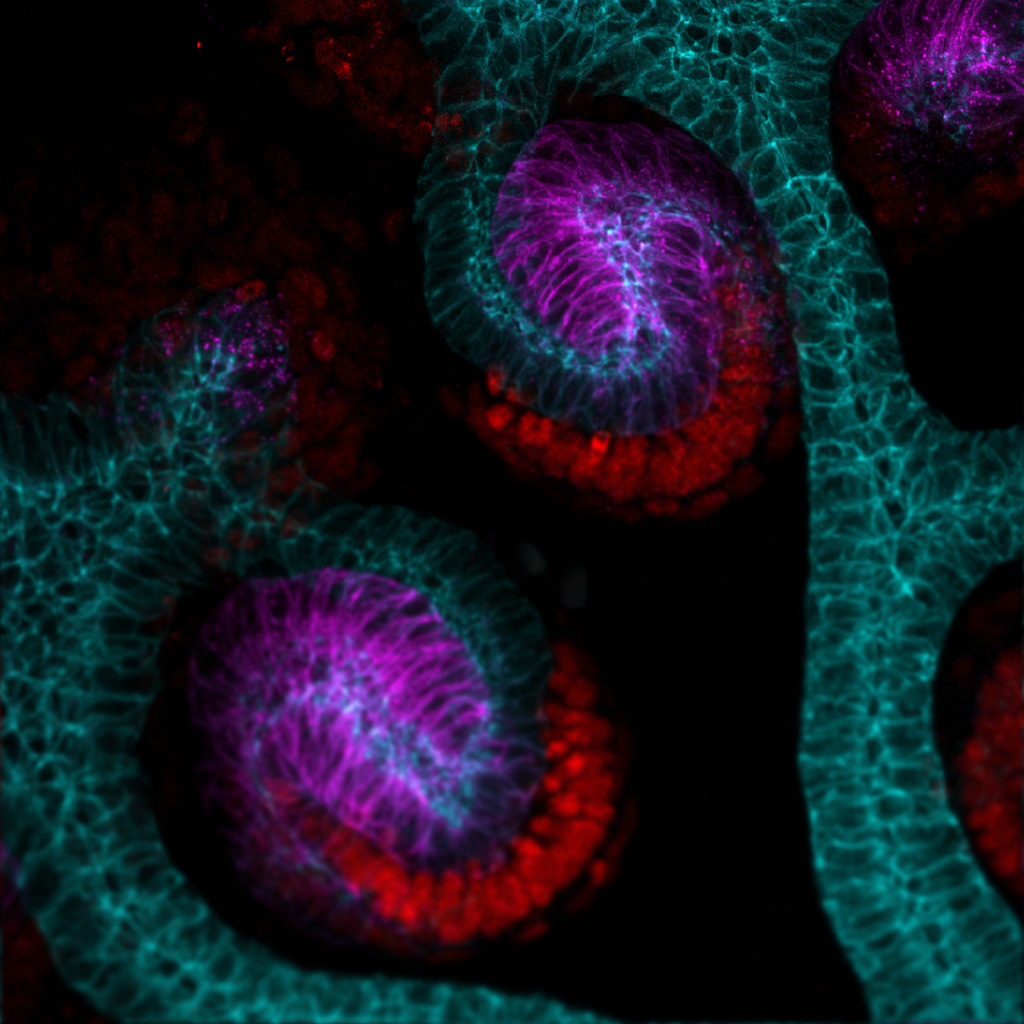} };
    \zoombox[color code=red]{0.4,0.4}
    \zoombox[color code=blue]{0.5,0.9}
    \zoombox[magnification=3]{0.1,0.1}
    \zoombox[color code=green]{0.9,0.5}
\end{tikzpicture}
\caption{Blurred with blur Figure \ref{fig:psf_blur_1} -- SNR = $25.30$ dB} \label{fig:confocal_blurred_1}
\end{subfigure}

\begin{subfigure}[b]{0.9\textwidth} \centering
\begin{tikzpicture}[zoomboxarray]
    \node [image node] { \includegraphics[width=0.4\textwidth]{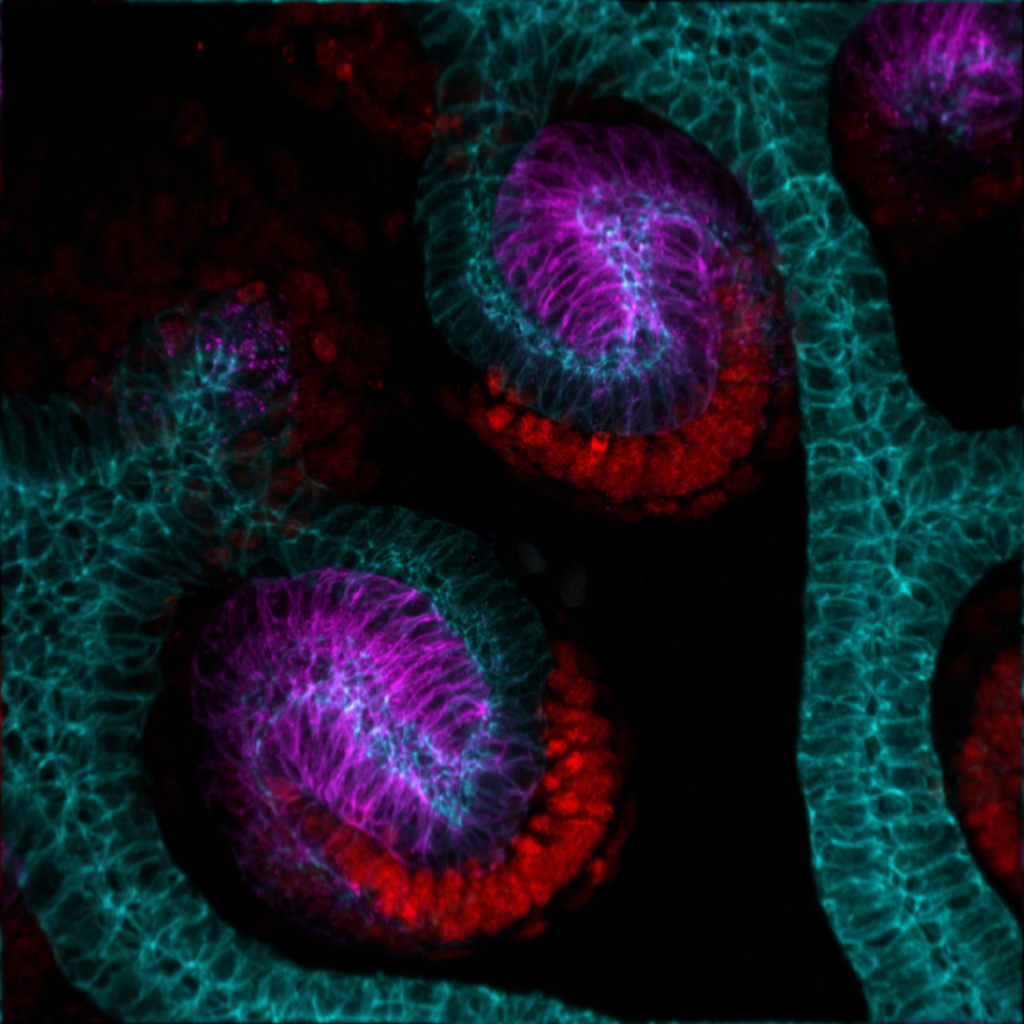} };
    \zoombox[color code=red]{0.4,0.4}
    \zoombox[color code=blue]{0.5,0.9}
    \zoombox[magnification=3]{0.1,0.1}
    \zoombox[color code=green]{0.9,0.5}
\end{tikzpicture}
\caption{Blurred with blur Figure \ref{fig:psf_blur_2} -- SNR = $22.72$ dB} \label{fig:confocal_blurred_2}
\end{subfigure}
  \caption{Images $f_0$, degraded versions of $f$ for the two operators considered. Noise level $\sigma = 5.10^{-3}$} \label{fig:images_original_degraded}
\end{figure}

\paragraph{The optimization problem}

In this experiment we will seek to recover the image $f$ by solving the following variational problem
\begin{equation} \label{eq:optim_deblurring}
    \argmin_{f \in \mathcal{E}} E(f) = \frac{1}{2} \left\| H f - f_0 \right\|_2^2 + \left\| \Psi^* f \right\|_{1,w},
\end{equation}
where $\| z \|_{1,w}$ is a weighted $\ell_1$-norm with weights $w \in \R^N$ and defined by
\begin{equation*}
    \| z \|_{1,w} = \sum_{\lambda \in \Lambda} w[\lambda] \left|z[\lambda] \right|. 
\end{equation*}
The solution of this variational problem does not correspond to the state-of-the-art in terms of image quality since the prior is simple, but it is known to perform well in short computation times. Figure \ref{fig:images_deblurred} shows the solutions of \eqref{eq:optim_deblurring} for the two blurs, with $\Psi$ being a Symmlet wavelet basis of order 6, since it offers a good compromise between computing times and visual quality of the results \cite{escande2018accelerating}. We use this basis in all the experiments.

\begin{figure}[htpb] \centering
\begin{subfigure}[b]{0.9\textwidth} \centering
\begin{tikzpicture}[zoomboxarray]
    \node [image node] { \includegraphics[width=0.4\textwidth]{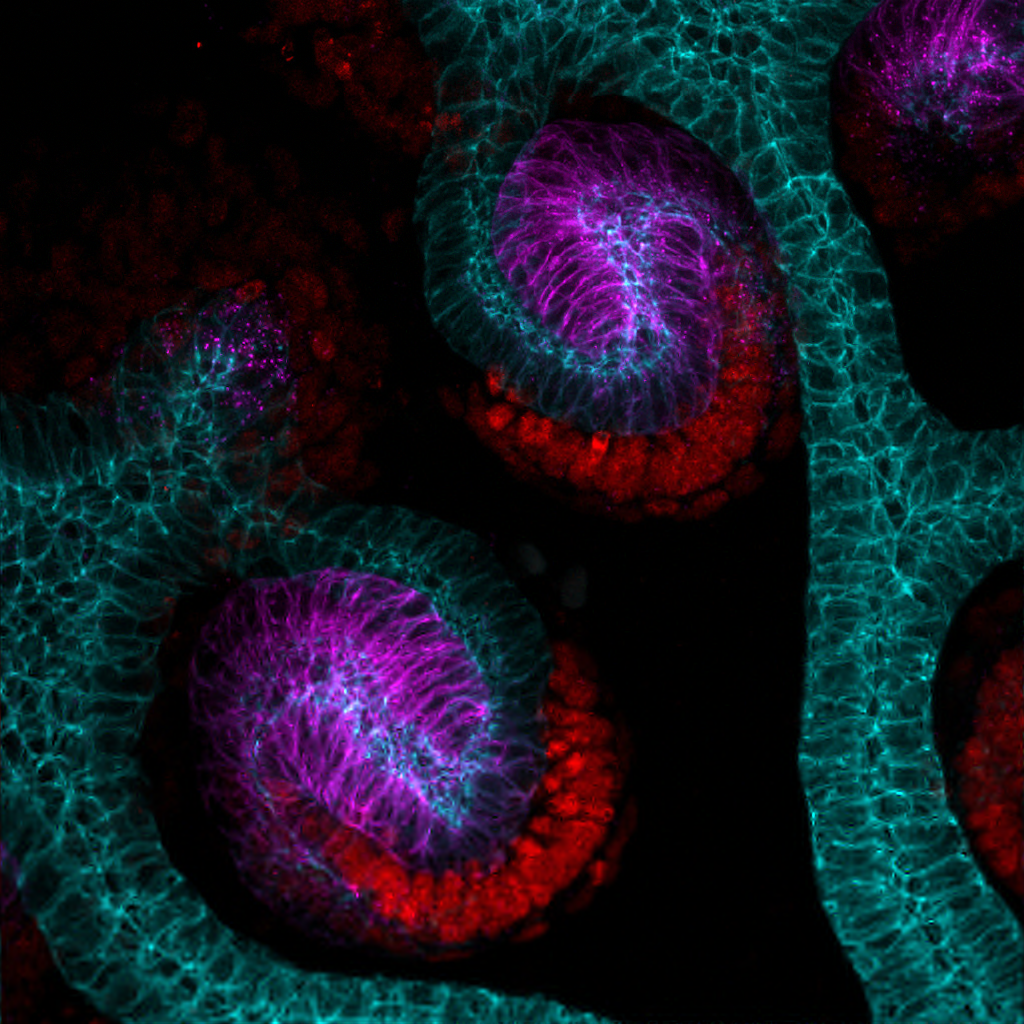} };
    \zoombox[color code=red]{0.4,0.4}
    \zoombox[color code=blue]{0.5,0.9}
    \zoombox[magnification=3]{0.1,0.1}
    \zoombox[color code=green]{0.9,0.5}
\end{tikzpicture}
\caption{Restored image -- SNR = $27.66$ dB} \label{fig:deblurred_1}
\end{subfigure}

\begin{subfigure}[b]{0.9\textwidth} \centering
\begin{tikzpicture}[zoomboxarray]
    \node [image node] { \includegraphics[width=0.4\textwidth]{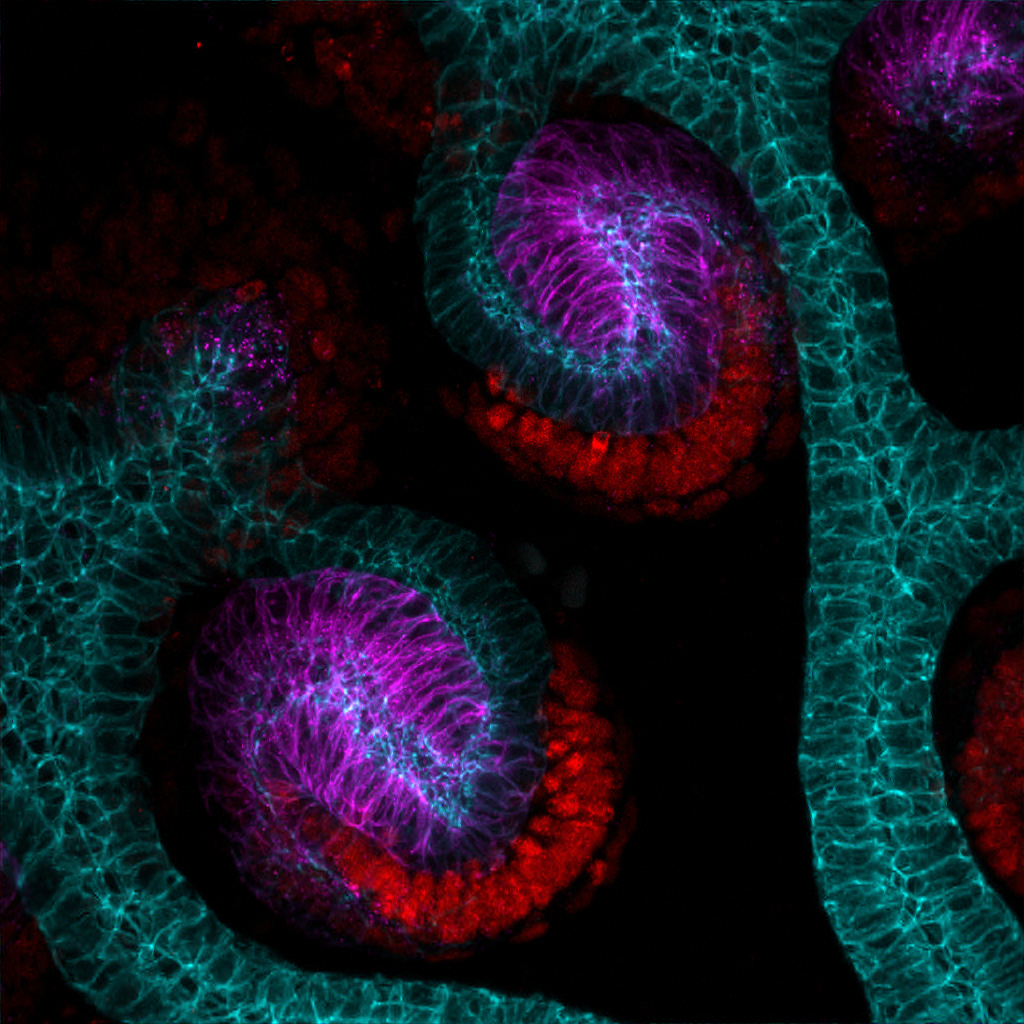} };
    \zoombox[color code=red]{0.4,0.4}
    \zoombox[color code=blue]{0.5,0.9}
    \zoombox[magnification=3]{0.1,0.1}
    \zoombox[color code=green]{0.9,0.5}
\end{tikzpicture}
\caption{Restored image -- SNR = $29.71$ dB} \label{fig:deblurred_2}
\end{subfigure}

  \caption{Solutions of \eqref{eq:optim_deblurring} for the two operators considered. The exact operator is used with $w[\lambda] = 2 \cdot 10^{-2} \cdot |\lambda|$.} \label{fig:images_deblurred}
\end{figure}

\paragraph{The proposed algorithms}

Many algorithms were designed to solve non-smooth convex problems of the form \eqref{eq:optim_deblurring}. The two predominant algorithms are the (accelerated) proximal gradient descent \cite{nesterov2013gradient} also known as FISTA \cite{beck2009fast} and the alternating descent method of multipliers (ADMM) \cite{glowinski2008lectures}.

As shown in \cite{escande2018accelerating}, a very efficient version of FISTA can be derived by approximating \eqref{eq:optim_deblurring} in the wavelet domain i.e.
\begin{equation} \label{eq:optim_deblurring_wavelet}
    \argmin_{z \in \R^N} \frac{1}{2} \left\| \Theta_L z - z_0 \right\|_2^2 + \left\| z \right\|_{1,w},
\end{equation}
where $z_0 = \Psi^* f_0$, and $\Theta_L$ is an $L$-sparse approximation of $\Theta = \Psi^* H \Psi$.
This idea allows to avoid constantly swapping between the wavelet and spatial domains, where most of the time is spent in traditional solvers. In addition, the matrix $\Theta_L^*\Theta_L$ - which appears in the gradient of the quadratic term - is mostly concentrated around its diagonal that decreases across sub-bands. An efficient diagonal preconditioner can be designed exploiting this property and further reduce the number of iterations by a factor roughly equal to $2$ without compromising image quality. 

The ADMM on its side requires the inversion of symmetric positive definite linear systems with matrix $\Theta_L^*\Theta_L+\tau \mathrm{I}_N$. This step could also be accelerated using the same ideas, but we observed that it was less efficient in practice and will not further report about this.

Table \ref{table:algos} compares the performance of various algorithms detailed in more details in Appendix \ref{sec:deblurring_algos}.

\begin{table}[htbp] \centering
\begin{tabular}{@{} p{0.12\textwidth} @{\hskip 0.07\textwidth} p{0.4\textwidth} @{\hskip 0.07\textwidth} p{0.06\textwidth} @{} p{0.06\textwidth} @{} p{0.06\textwidth} @{} p{0.06\textwidth} @{}}
    \toprule
    Algorithm & Description & \multicolumn{4}{c}{Complexity} \\   
    & & {\footnotesize MVP} & {\footnotesize FFT} & {\footnotesize FWT} & {\footnotesize Diag} \\ \midrule 
    {\footnotesize FISTA-Spa} & FISTA with the exact operator & {\footnotesize 2} & {\footnotesize -} & {\footnotesize 2} & {\footnotesize -} \\[0.2cm]
    {\footnotesize FISTA-PC} & FISTA using a PC expansion & {\footnotesize -} & {\footnotesize $2m$} & {\footnotesize 2} & {\footnotesize $2m$} \\[0.2cm]
    {\footnotesize FISTA-W} & FISTA with sparse $\Theta_L$ & {\footnotesize 2} & {\footnotesize -} & {\footnotesize -} &  {\footnotesize -}\\[0.2cm]
    {\footnotesize FISTA-WP} & FISTA with $\Theta_L$ and Jacobi preconditioner & {\footnotesize 2} & {\footnotesize -} & {\footnotesize -} & {\footnotesize 1} \\[0.2cm]
    {\footnotesize ADMM-PC} & See Appendix \ref{sec:admm} & {\footnotesize -} & {\footnotesize $3m$} & {\footnotesize 3} & {\footnotesize $6m+2$}\\
    \bottomrule
\end{tabular}
\caption{List of algorithms compared in the numerical experiments with their complexity per iteration in terms of matrix vector products (MVP), fast Fourier transforms (FFT), wavelet transforms (WT) and product with diagonal matrices (Diag).} \label{table:algos}
\end{table}

\subsubsection{Choice of \texorpdfstring{$m$}{m} and \texorpdfstring{$\eta$}{precision}} \label{sec:choice_m}

The approximation properties of product-convolution expansions are investigated in \cite{escande2016approximation} in terms of Frobenius distance on the operators.
We are dealing here with an inverse problem thus the approximation needs are different. 
Because of the regularization term in \eqref{eq:optim_deblurring}, fine approximations of the operator are unnecessary, allowing to reduce the computational burden, see Figure \ref{fig:choice_m}. 

We choose the order of expansion $m$ in a such a way that the pSNR of the deblurred images does not differ from more than $0.01$dB from the pSNR of the solution of \eqref{eq:optim_deblurring}. 
From Figure \ref{fig:choice_m}, we obtain $m = 5$ for the operator on Figure \ref{fig:psf_blur_1} and $m = 25$ for the operator on Figure \ref{fig:psf_blur_2}. Note that $m$ is larger for the second blur since the impulse response variations are more complex. 

\begin{figure}[htpb] \centering
\begin{subfigure}[b]{0.45\textwidth}
%
%
\definecolor{mycolor1}{rgb}{0.00000,0.44700,0.74100}%
\definecolor{mycolor2}{rgb}{0.85000,0.32500,0.09800}%
\begin{tikzpicture}

\begin{axis}[%
width=0.856\textwidth,
height=0.9\textwidth,
at={(0\textwidth,0\textwidth)},
scale only axis,
xmin=1,
xmax=10,
ymin=12,
ymax=30,
xlabel=$m$,
axis background/.style={fill=white},
legend style={legend cell align=left,align=left,draw=white!15!black},
legend pos=south east,
]
\addplot [color=mycolor1,solid, thick,mark=x]
  table[row sep=crcr]{%
1 13.584934 \\ 
2 25.012203 \\ 
3 27.260059 \\ 
4 27.621129 \\ 
5 27.660344 \\ 
6 27.665037 \\ 
7 27.664931 \\ 
8 27.664620 \\ 
9 27.664756 \\ 
10 27.664868 \\ 
11 27.664630 \\ 
12 27.664575 \\ 
13 27.664579 \\ 
14 27.664598 \\ 
15 27.664561 \\ 
16 27.664562 \\ 
17 27.664579 \\ 
18 27.664559 \\ 
19 27.664573 \\ 
20 27.664560 \\ 
21 27.664575 \\ 
22 27.664562 \\ 
23 27.664585 \\ 
24 27.664558 \\ 
25 27.664571 \\ 
26 27.664566 \\ 
27 27.664543 \\ 
28 27.664558 \\ 
29 27.664565 \\ 
30 27.664570 \\ 
};
\addlegendentry{pSNR PC};

\addplot [color=mycolor2,dashed,very thick]
  table[row sep=crcr]{%
1	27.664565161284\\
30	27.664565161284\\
};
\addlegendentry{pSNR exact};

\end{axis}
\end{tikzpicture}%
\end{subfigure}
$\qquad$
\begin{subfigure}[b]{0.45\textwidth}
%
%
\definecolor{mycolor1}{rgb}{0.00000,0.44700,0.74100}%
\definecolor{mycolor2}{rgb}{0.85000,0.32500,0.09800}%
\begin{tikzpicture}

\begin{axis}[%
width=0.856\textwidth,
height=0.9\textwidth,
at={(0\textwidth,0\textwidth)},
scale only axis,
xmin=5,
xmax=30,
ymin=25,
ymax=30,
xlabel=$m$,
axis background/.style={fill=white},
legend style={legend cell align=left,align=left,draw=white!15!black},
legend pos=south east
]
\addplot [color=mycolor1,solid, thick,mark=x]
  table[row sep=crcr]{%
1 17.663160 \\ 
2 19.665483 \\ 
3 20.765925 \\ 
4 24.484914 \\ 
5 24.992414 \\ 
6 25.387114 \\ 
7 25.474954 \\ 
8 28.491673 \\ 
9 28.627515 \\ 
10 28.833819 \\ 
11 29.005854 \\ 
12 29.028120 \\ 
13 29.070750 \\ 
14 29.076217 \\ 
15 29.090793 \\ 
16 29.139465 \\ 
17 29.174553 \\ 
18 29.550345 \\ 
19 29.559565 \\ 
20 29.565745 \\ 
21 29.568412 \\ 
22 29.573239 \\ 
23 29.580070 \\ 
24 29.693663 \\ 
25 29.703000 \\ 
26 29.703220 \\ 
27 29.706505 \\ 
28 29.704728 \\ 
29 29.714160 \\ 
30 29.715642 \\ 
};
\addlegendentry{pSNR PC};

\addplot [color=mycolor2, dashed, very thick]
  table[row sep=crcr]{%
1	29.7156422282043\\
30	29.7156422282043\\
};
\addlegendentry{pSNR exact};

\end{axis}
\end{tikzpicture}%
\end{subfigure}
  \caption{ pSNR of the solution of \eqref{eq:optim_deblurring} obtained with the product-convolution expansion of order $m$ (blue) and the exact operator (red). Left: for the operator on Figure \ref{fig:psf_blur_1}. Right: for the operator on Figure \ref{fig:psf_blur_2}.} \label{fig:choice_m}
\end{figure}
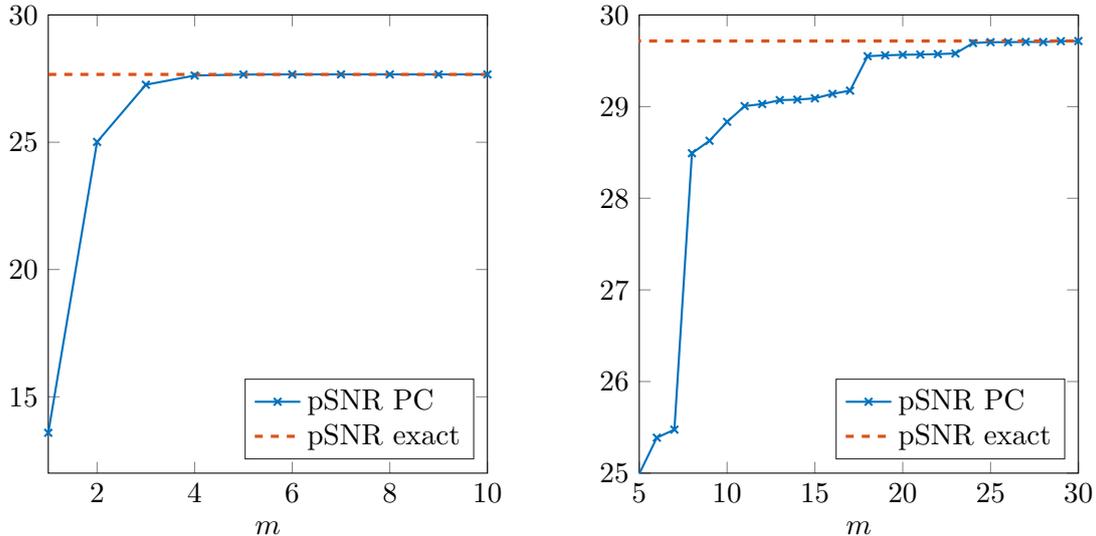

In a similar fashion, we select the precision $\eta$ in such a way that the pSNR of the solution of \eqref{eq:optim_deblurring_wavelet} does not differ from more that $0.04$dB from the pSNR of the solution of \eqref{eq:optim_deblurring}. We do not report the plots, but we found that $\eta = 5.10^{-4}$ was sufficient to ensure this requirement.

\subsubsection{Computing times}

In this paragraph, we precisely compare the different methods designed to solve \eqref{eq:optim_deblurring}. The methodology consists in:
\begin{itemize}
    \item finding the number $L$ for which the solution of \eqref{eq:optim_deblurring_wavelet} has a decrease of pSNR of less than 0.2dB w.r.t. the pSNR of the solution of \eqref{eq:optim_deblurring},
    \item for each algorithm, finding a number of iteration $Nit$ leading to a precision.
    \begin{equation} \label{eq:stopping_crit}
        E(f^{(Nit)}) - E(f^{(0)}) \leq 10^{-3} E(f^{(0)}).
    \end{equation}
\end{itemize}

The cost functions are reported in Figure \ref{fig:deblurring_comparisons}, while the timings are reported in Tables \ref{table:timings_1} and \ref{table:timings_2}. 
Deblurring were performed on Matlab with automatic multi-threading disabled (using the -singleCompThreah option), to avoid uncontrolled behaviors while timing. 
To further demonstrate the efficiency of FISTA-WP, we implement it in CUDA and run it on a GeForce GTX Titan X in double precision.

\begin{figure}[htpb]
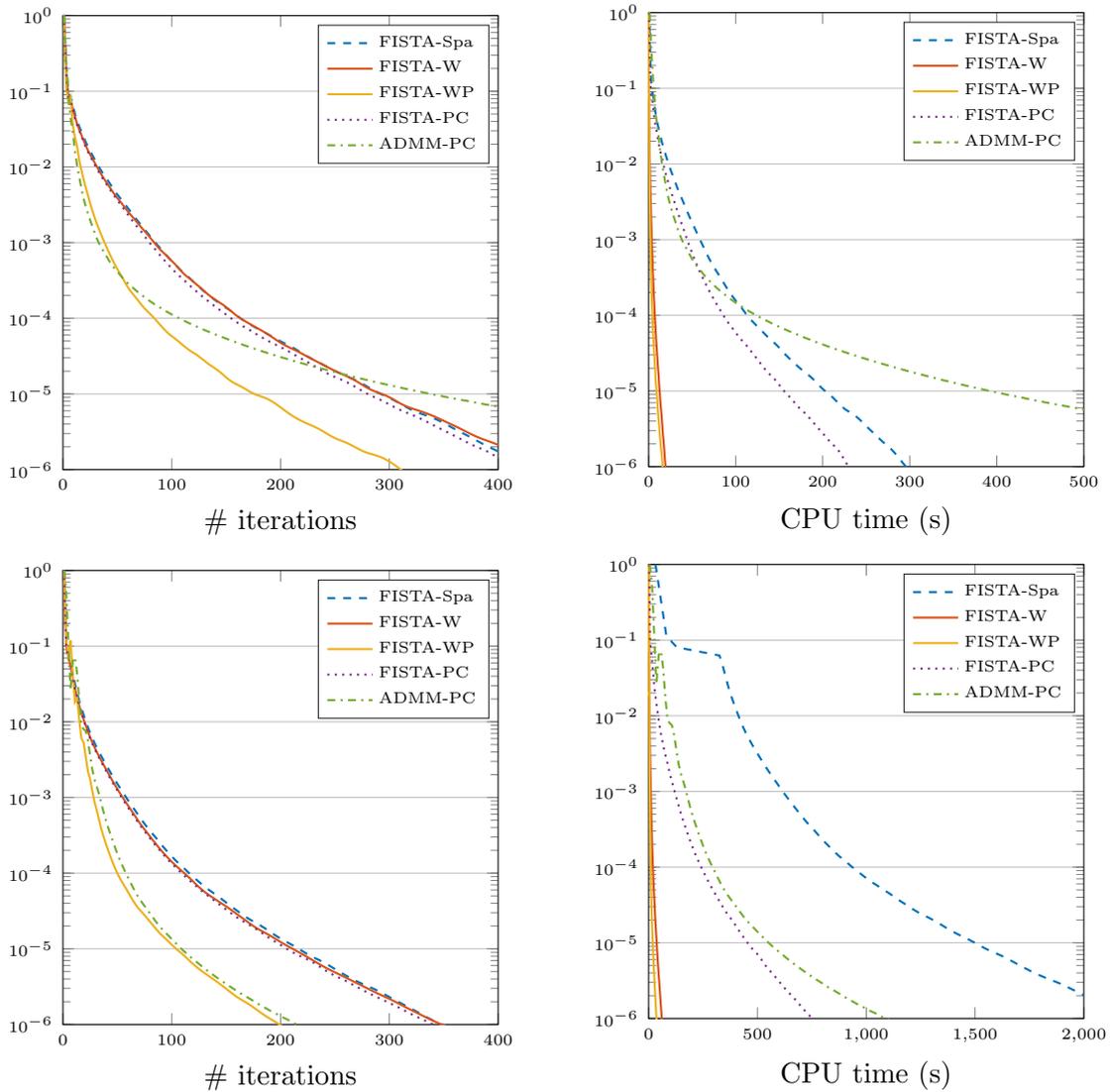
 \centering
\begin{subfigure}[b]{0.45\textwidth}
    \input{images/deblurring/comparisons_6/deblurring_cf_vs_it.tex}
\end{subfigure}
$\qquad$
\begin{subfigure}[b]{0.45\textwidth}
    \input{images/deblurring/comparisons_6/deblurring_cf_vs_time.tex}
\end{subfigure}

\begin{subfigure}[b]{0.45\textwidth}
    \input{images/deblurring/comparisons_rotation_skew/deblurring_cf_vs_it.tex}
\end{subfigure}
$\qquad$
\begin{subfigure}[b]{0.45\textwidth}
    \input{images/deblurring/comparisons_rotation_skew/deblurring_cf_vs_time.tex}
\end{subfigure}
  \caption{Performance of the deblurring methods in Table \ref{table:algos} for an image of size $1024\times 1024$. The cost functions are displayed with respect to the number of iterations on the left column and the time on the right column. The first row corresponds to the blur on Figure \ref{fig:psf_blur_1} and the second one to the blur on Figure \ref{fig:psf_blur_2}.} \label{fig:deblurring_comparisons}
\end{figure}

\begin{table}[htbp] \centering
{\small
    \begin{tabular}{@{}l c c c c c c @{}}
        \toprule
            & FISTA-Spa & FISTA-PC & ADMM-PC & FISTA-W & FISTA-WP & FISTA-WP (GPU) \\ \midrule
        \# iterations & 84 & 80 & 32 & 84 & 38 & 38 \\
        Time (s) & 58.3 & 44.4 & 37.6 & 3.6 & 1.9 & 0.087\\
        Speed-up & - & 1.3 & 1.5 & 16.2 & 30.7 & 670 \\
        \bottomrule
    \end{tabular}}
    \caption{Timings and iterations number to reach criterion \eqref{eq:stopping_crit} for the blur in Figure \ref{fig:psf_blur_1} for the algorithms in Table \ref{table:algos} and in Figure\ref{fig:confocal} ($N = 1024\times 1024$).} \label{table:timings_1}
\end{table}

\begin{table}[htbp] \centering
{\small
    \begin{tabular}{@{}l c c c c c c@{}}
        \toprule
        & FISTA-Spa & FISTA-PC & ADMM-PC & FISTA-W & FISTA-WP & FISTA-WP (GPU) \\ \midrule        
        \# iterations & 58 & 54 & 34 & 54 & 28 & 28 \\
        Time (s) & 616 & 119.1 & 172.6 & 9.9 & 5.1 & 0.164\\
        Speed-up & - & 5.2 & 3.6 & 62 & 120 & 4089 \\
        \bottomrule
    \end{tabular}}
    \caption{Timings and iterations number to reach criterion \ref{eq:stopping_crit} for the blur in Figure \ref{fig:psf_blur_2} for the algorithms in Table \ref{table:algos} and in Figure\ref{fig:confocal} ($N = 1024\times 1024$).}\label{table:timings_2}
\end{table}

This example confirms that solving problem \eqref{eq:optim_deblurring} by an approximation in the wavelet domain \eqref{eq:optim_deblurring_wavelet} leads to dramatic reduction of computation times ($\times 16$ and $\times 62$ speed-up factors) with little loss in the quality (less than 0.2dB). 
Furthermore, the use of the structure of $\Theta_L$ allows to derive efficient diagonal preconditioners which lead to an extra $\times 2$ speed-up.

\section{Proof of the main result} \label{sec:proofs}


The proof exploits the peculiar structure of product-convolution expansions:
\begin{itemize}
 \item the convolution operator $U_k = u_k \star \cdot$ can be decomposed in the wavelet basis and yields a matrix $A_k = \Psi^* U_k \Psi$ in $O(N \log_2^2 N)$ operations. As long as $U_k$ belongs to the smoothness class described in Definition \ref{def:smoothness}, the matrix $A_k$ can be thresholded to obtain a sparse approximation $\widetilde{A}_k$  of $A_k$.
 \item the multipliers $V_k = v_k \odot \cdot$ are diagonal. Therefore the sparsity pattern of $B_k = \Psi^* V_k \Psi$ is included in the set of wavelet with supports intersecting the diagonal and contain at most $O(N \log_2 N)$ non-zero coefficients, with known locations. They can be computed in $O(N \log_2 N)$ operations using a cascade algorithm.
 \end{itemize}
The general approach is therefore to compute each couple $(A_k, B_k)$, compute the products $\widetilde{A}_k B_k$ and accumulate the sum to get an approximation $\widetilde \Theta = \sum_{k=1}^m \widetilde{A}_k  B_k$ of $\Theta$.
The last two steps are also critical since in the general case, the number of coefficients can explode in the product $\widetilde{A}_k B_k$ or when summing them all.
However, relying again on the peculiar structure of the matrices $\widetilde{A}_k$ and $B_k$, one can prove that the complexity of the algorithm stays quasi-linear. In what follows we precisely describe each of these steps.

\subsection{Structure of \texorpdfstring{$A_k$}{A}}

The matrix $A_k = \Psi U_k \Psi^*$ is the wavelet representations of a convolution operator $U_k$. 
As such, it inherits a peculiar structure. 
As observed in \cite{escande2018accelerating} it has circulant sub-bands. 
To make it precise, we recall that a wavelet representation $\Theta$ can be decomposed into its wavelet sub-bands:
\begin{equation*}
	\Theta = \left( \Theta_{j,j'}^{e,e'}\right)_{j,j' \leq J-1}^{e,e' \in \{0,1\}^d}, \quad \textrm{with} \, \Theta_{j,j'}^{e,e'}[l,l'] = \Theta[\lambda,\mu].  
\end{equation*}
for $l \in \mathcal{T}_j$, $l' \in \mathcal{T}_{j'}$ and $\lambda = (j,e,l)$, $\mu = (j',e',l')$. 
For instance on 1D signals with $J=2$, the sub-bands of $\Theta$ are shown in Figure \ref{fig:struct_theta}.

\begin{figure}[htp] \centering
	\includegraphics[width=0.25\textwidth]{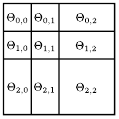}
	\caption{Sub-band structure of convolution matrices in the wavelet domain. Here on 1D signals with $J=2$.} \label{fig:struct_theta}
\end{figure}

The sub-bands  $\Theta_{j,k}^{e,e'}$  themselves have a specific structure captured by the following definition.

\begin{definition}[Rectangular circulant matrices]
	Let $P \in \R^{2^{jd} \times 2^{kd}}$ denote a rectangular matrix and $\tau_a$ be the translation operator by vector $a \in \Omega$.
	The matrix $P$ is called circulant if and only if
	\begin{itemize}
		\item When $k \geq j$ : there exists $p \in \R^{2^{kd}}$ such that $P[l,:] = \tau_{2^{k-j}l}{p}$ for all $l \in \mathcal{T}_j$.
		\item When $k < j$ : there exists $p \in \R^{2^{jd}}$ such that $P[:,l] = \tau_{2^{j-k}l}{p}$ for all $l \in \mathcal{T}_k$.
	\end{itemize}
\end{definition}

\begin{lemma}[{\cite[Theorem 3]{escande2018accelerating}}] \label{lem:struct_A}
	The sub-bands of matrix $A_k$ are circulant. Therefore, only $O\left((2^{d}-1) N \log_2 N\right)$ coefficients are needed to encode $A_k$. They can be computed with no more than $O\left((2^{d}-1) N \log_2^2 N\right)$ operations.
\end{lemma}

Furthermore, the matrices $A_k$ are assumed to belong to the smoothness class $\mathcal{A}_\alpha$ governing the decay of its coefficients.
This properties allows to obtain efficient approximations of $A_k$ as described by Theorem \ref{thm:cohen}.
\begin{theorem}[{\cite[Theorem 4.6.2]{cohen2003numerical}}] \label{thm:cohen}
	Let $A \in \mathcal{A}_{\alpha}$. For all $L \geq 0 $, one can construct a matrix $A_L$ with at most $L$ non-zero entries of $A$ per row and column such that 
	\begin{equation}
		\| A - A_L \|_{2 \to 2} \lesssim L^{-\alpha/d}.
	\end{equation}
\end{theorem}

As a summary, Lemma \ref{lem:struct_A} together with Theorem \ref{thm:cohen} ensure that the computation of the matrices $A_k$ and their approximations at any precision are tractable in large dimensions.

\begin{remark}
	Note that the approximation step is crucial to obtain a quasi-linear complexity of the overall algorithm. 
	If no approximation were made, the complexity of the product $A_k B_k$ could become quasi-quadratic.
\end{remark}

\begin{remark}
	The proof of Theorem \ref{thm:cohen} gives an explicit way of building $A_L$ from $A$. We recall it here for sake of completeness. The matrix $A_L$ is obtained by zeroing all entries satisfying either (i) $| |\lambda| - |\mu| | > J(L)$ or (ii) $\vartheta(\lambda,\mu) > 2^{J(L)}$ where $J(L)$ is linked to $L$ through the relation $L \lesssim J(L) 2^{dJ(L)}$. A straightforward algorithm building $A_L$ consists in looping along all rows $\lambda$ and rows $\mu$ and discard elements satisfying either (i) or (ii). This algorithm requires $O(N^2)$ operations in the general case. Due to the particular structure of matrix $A$ - only $\log_2(N)$ coefficient are non-zero per row (Lemma \ref{lem:struct_A}) - this algorithm has $O(N \log_2 N)$ complexity. 
	In practice, many authors propose to simply threshold $A$. This comes from the observation that keeping coefficients that are larger than $2^{-(d/2+\alpha) J}$ do not satisfy (i) and (ii) (from Definition \ref{def:smoothness}). Again, since $A$ is encoded with $O(N\log_2N)$ coefficients, thresholding $A$ costs $O(N\log_2N)$ instead of $O(N^2)$ in the general case.
\end{remark}

\subsection{Structure of \texorpdfstring{$B_k$}{B}}

We now turn on discussing the structures of the matrices $B_k = \Psi V_k \Psi^*$. Since the operators $V_k$ are diagonal, the coefficients $B_k[\lambda,\mu] = \langle v_k \odot \psi_\lambda, \psi_\mu \rangle$ necessarily vanish as long as $\supp \psi_\lambda$ and $\supp \psi_\mu$ do not intersect.
Define 
\begin{equation*}
	\mathcal{I} = \left\{ (\lambda,\mu) \in \Lambda \times \Lambda \, | \, \supp \psi_\lambda \cap \supp \psi_\mu \neq \emptyset \right\},
\end{equation*}
we thus have $\supp B_k \subseteq \mathcal{I}$ for all $1 \leq k \leq m$.
This property allows to derive upper bounds for the number of coefficients in a column and a row of $B_k$ (Lemma \ref{lem:num_coeff_row_B}) as well as for the cardinality of $\supp B_k$ (Lemma \ref{lem:num_intersections}).

\begin{lemma} \label{lem:num_coeff_row_B}
    Assume that the mother wavelet is compactly supported on a hypercube of sidelength $2\delta$.
	Let $\lambda \in \Lambda$ and $j=|\lambda|$, there are at most $(2\delta)^d \left( (2^d-1) j + 2^{d(J-j)}\right)$ indexes $\mu \in \Lambda$ such that $\supp \psi_\lambda \cap \supp \psi_\mu \neq \emptyset$. 
\end{lemma}
\begin{proof}
	Let $\lambda = (j,e,l) \in \Lambda$.
	From the definition of $\psi_\lambda$ we get that 
	\begin{equation*}
		\supp \psi_\lambda \subseteq \left[-(2^{J-j}-1)(\delta-1)+2^{J-j}l,(2^{J-j}-1)\delta+2^{J-j}l\right].
	\end{equation*}
	Consider now another $\mu = (j',e',l') \in \Lambda$. Necessary conditions for the intersections of the supports of $\psi_\lambda$ and $\psi_\mu$ are
	\begin{equation*}
		\begin{split}
		-(2^{J-j'}-1)(\delta-1) + 2^{J-j'}l' & \leq (2^{J-j}-1)\delta + 2^{J-j}l \\
		-(2^{J-j}-1)(\delta-1) + 2^{J-j}l & \leq (2^{J-j'}-1)\delta + 2^{J-j'}l',
		\end{split}
	\end{equation*}
	where the $\leq$ has to be understood as a component-wise inequality for $d>1$.
	Without loss of generality, we assume that $l = 0$. 

	We let $\mathcal{N}_{j,j'}$ be the set of $l' \in \Lambda$ satisfying the above conditions i.e.
	\begin{equation*}
		\left\{
		\begin{aligned}
			l' & \geq -2^{j'-J}\left( 2^{J-j} + 2^{J-j'} - 2 \right) \delta + 2^{j'-J}(2^{J-j}-1) \\
			l' & \leq 2^{j'-J}\left( 2^{J-j} + 2^{J-j'} -2\right) \delta - 2^{j'-J}(2^{J-j'}-1)
		\end{aligned}
		\right.
	\end{equation*}

	The volume of $\mathcal{N}_{j,j'}$ is therefore 
	\begin{equation*}
		\begin{split}
		\# \mathcal{N}_{j,j'} & = \left( 2^{j'-J} \left( 2^{J-j} + 2^{J-j'} - 2 \right) (2 \delta-1) + 1 \right)^d \\
							& \leq \left( 2^{j'} \left( 2^{-j} + 2^{-j'} \right) (2 \delta-1) + 1 \right)^d \\
							& \leq \left( 2^{j'} \left( 2^{-j} + 2^{-j'} \right) 2 \delta + 1 - 2^{j'-j} -1  \right)^d \\
							& \leq \left( 2^{j'} \left( 2^{-j} + 2^{-j'} \right) 2 \delta  \right)^d 
		\end{split}
	\end{equation*}
	
	Therefore,
	\begin{equation*}
	\begin{split}
	\# &\left\{ \mu \in \Lambda \, | \, \supp \psi_\mu \cap \supp \psi_\lambda \neq \emptyset \right\} \\
	& \leq \sum_{j'=0}^{J-1} \sum_{e \in \{0,1\}^d \setminus \{0\}} \sum_{l' \in \mathcal{T}_{j'}} \indic_{\mathcal{N}_{j,j'}}(l') \\
	& = (2^d-1)\sum_{j'=0}^{J-1} \left( 2^{j'} \left( 2^{-j} + 2^{-j'} \right) 2 \delta  \right)^d  \\
	& \leq (2^d-1) (2\delta)^d \left( \sum_{j'=0}^{j-1} 1 + \sum_{j'=j}^{J-1} 2^{d(j'-j)} \right) \\
	& = (2^d-1) (2\delta)^d \left( j + 2^{-jd} \frac{2^{dJ} - 2^{dj}}{2^d-1} \right) \\
	& \leq (2\delta)^d \left( (2^d-1) j + 2^{d(J-j)} \right)
	\end{split}
	\end{equation*}
	which gives the upper-bound on the number of intersections.
\end{proof}

\begin{lemma} \label{lem:num_intersections}
	We have $\# \mathcal{I} \leq c(d) \delta^d J 2^{dJ}$, where the constant $c(d) = 2(2^d-1) 2^{d}$.
	Therefore since $J = \frac{1}{d}\log_2(N) $, the total number of intersections is at most 
	\begin{equation*}
		\frac{c(d)}{d} \delta^d N \log_2(N).
	\end{equation*}
\end{lemma}
\begin{proof}
\begin{equation*}
	\begin{split}
	\# \mathcal{I} & = \#\left\{ (\lambda,\mu) \in \Lambda \times \Lambda \, | \, \supp \psi_\mu \cap \supp \psi_\lambda \neq \emptyset \right\} \\
	& \leq \sum_{j=0}^{J-1} (2^d-1) 2^{jd}  2^{d} \delta^d \left( (2^d-1) j + 2^{d(J-j)}\right)  \\
	& \leq (2^d-1) 2^d \delta^d \left( (2^d-1) \sum_{j=0}^{J-1} j 2^{jd} + \sum_{j=0}^{J-1} 2^{dJ}\right)  \\
	& \leq (2^d-1) 2^{d} \delta^d \left( (2^d-1) J \frac{2^{dJ} - 1}{2^d - 1} + J 2^{dJ}\right) \\
	& \leq 2(2^d-1) 2^{d} \delta^d J 2^{dJ} \\ 
	\end{split}
\end{equation*}
\end{proof}

The numerical computation of $B_k = \Psi^* V_k \Psi$ is conceptually simple. It consists in two steps:
\begin{description}
 \item[Step 1] Compute the wavelet transform of each column of $V_k$ i.e. build a matrix $\widehat V_k$ such that $\widehat V_k[\lambda,i] = \langle V_k[\cdot, i], \psi_\lambda \rangle$ for all $1 \leq i \leq N$ and $\lambda \in \Lambda$. 
Each column of $V_k$ is a discrete Dirac mass, i.e. contains only one element. 
Therefore, its wavelet transform can be computed with a minimal number of operations using a cascade algorithm taking advantage of the signals sparsity.
A cascade algorithm consists in applying discrete convolutions with filters $h$ and $g$ and subsampling the outputs in a recursive fashion. It is derived from the multiresolution structure of a wavelet transform. We refer to \cite[Section 7.3]{Mallat-Book} for more details.
 \item[Step 2] Compute the wavelet transform of each row of $\widehat V_k$ i.e. $B_k[\lambda,\mu] = \langle \widehat V_k[\lambda, \cdot], \psi_\mu \rangle$ for all $\lambda, \mu \in \Lambda$. 
Similarly, the $\lambda$-th row of $\widehat V_k$ is sparse: it contains at most $ (2^{J-j} 2 \delta)^d$ contiguous coefficients. The same sparse cascade algorithm can be used to compute its wavelet transform.
\end{description}
We illustrate the sparsity structure of the matrices $\widehat V_k$ and $B_k$ in Figure \ref{fig:structure_matrix_B} for the case of 1D signals. Observe that each row of $\widehat V_k$ contains only a few contiguous nonzero indexes.

\begin{figure}[htpb] \centering
\begin{subfigure}[b]{0.45\textwidth}
\centering
    \includegraphics[height=5cm]{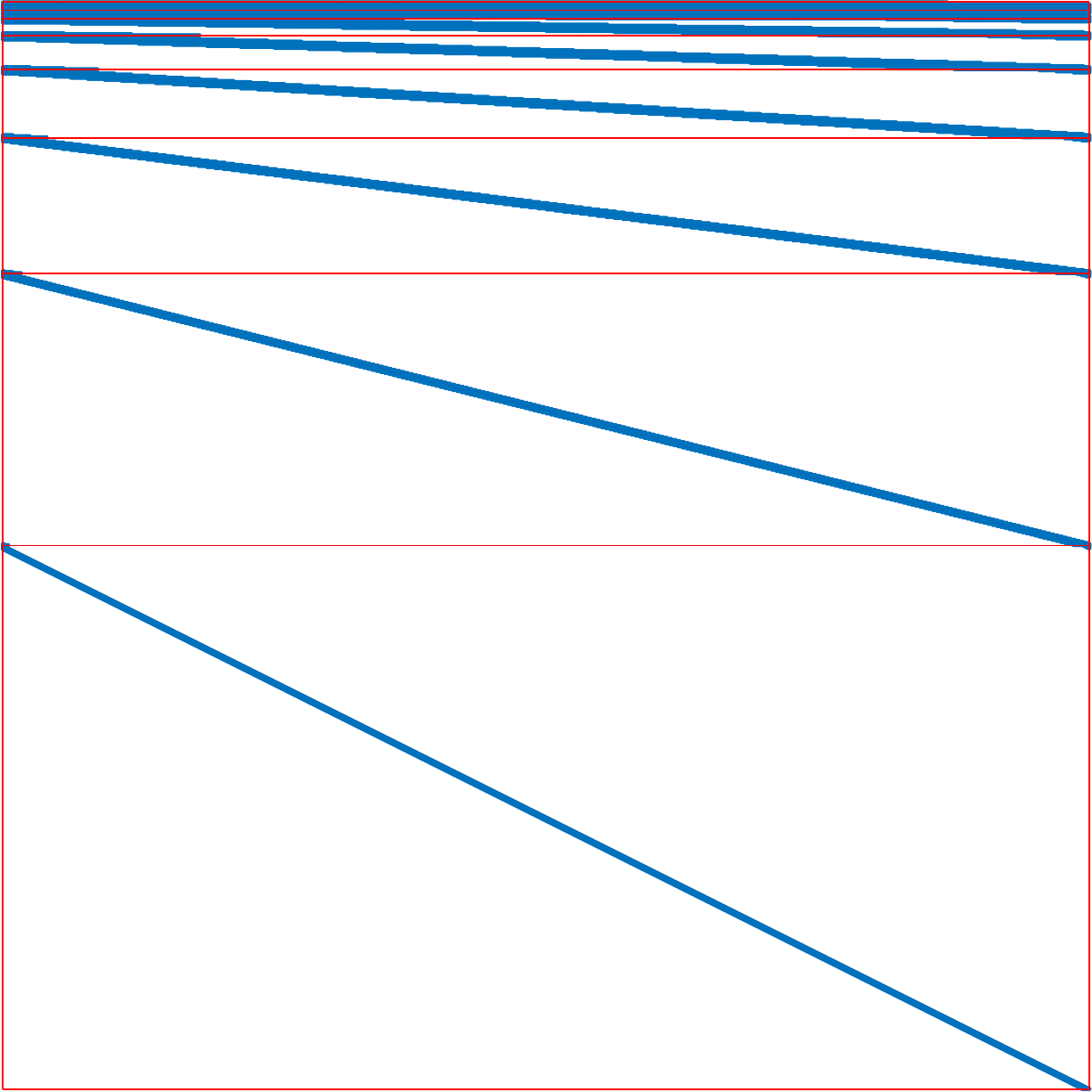}
    \caption{Sparsity structure of $\widehat V_k$}
\end{subfigure}
$\qquad$
\begin{subfigure}[b]{0.45\textwidth}
\centering
    \includegraphics[height=5cm]{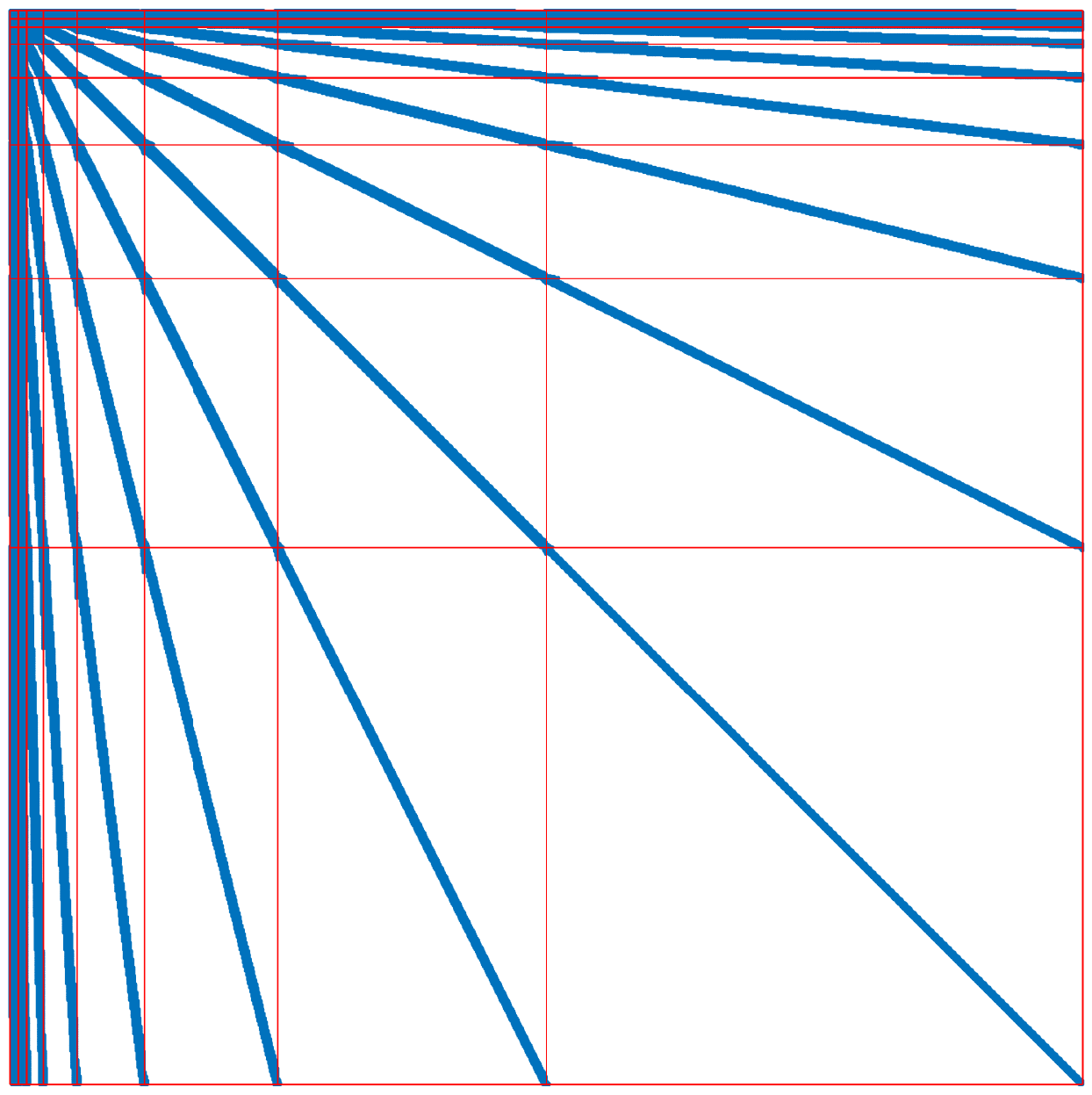}
    \caption{Sparsity structure of $B_k$}
\end{subfigure}
  \caption{Wavelet representations of a diagonal matrix. Here we illustrate the case of a 1D wavelet transform of size $1024\times 1024$ with $J=3$ decomposition levels.} \label{fig:structure_matrix_B}
\end{figure}

To control the complexity of each step, we will use the following result. 
\begin{lemma} \label{lem:wavelet_transform}
	Let $f \in \mathcal{E}$ be supported on the hypercube of $\Omega$ defined by the lower vertex $q_0 \in \Omega$ and the side length $\kappa$. We call $q_1 \in \Omega$ the upper vertex of the hypercube i.e. $q_1 = q_0 + \kappa$. 
	Let $c = \Psi^* f$ be its wavelet transform and $c_{j,e}$ its $(j,e)$-sub-band i.e. $c_{j,e}[l] = c[\lambda] = \langle f, \psi_\lambda \rangle$ for all $l \in \mathcal{T}_j$ and $\lambda = (j,e,l)$.
	We have
	\begin{equation*}
		\begin{aligned}
		\supp c_{j,e} & = \left[ 2^{j-J} q_0 - (1-2^{-J}), 2^{j-J}q_1+ (1-2^{-J})2(\delta-1) \right] \\
		\#\supp c_{j,e}  & \leq 2^d \left( 2^{d(j-J)} \kappa^d + (2\delta)^d \right)
		\end{aligned}
	\end{equation*}
	Furthermore, the number of operations to compute $c$ is bounded above by 
	\begin{equation*}
		2 d \delta 2^d \left(  \kappa^d +  (2^d-1) (2 \delta)^d J \right).
	\end{equation*}
\end{lemma}

\begin{proof}
	Let $\lambda = (j,e,l) \in \Lambda$.
	From the definition of $\psi_\lambda$ we get that 
	\begin{equation*}
		\supp \psi_\lambda \subseteq \left[-(2^{J-j}-1)(\delta-1)+2^{J-j}l,(2^{J-j}-1)\delta+2^{J-j}l\right].
	\end{equation*} 
	Necessary conditions for the intersections of the supports of $\psi_\lambda$ and $f$ are
	\begin{equation*}
	\left\{
		\begin{aligned}
		q_0 & \leq (2^{J-j}-1)\delta + 2^{J-j}l \\
		q_1 & \geq -(2^{J-j}-1)(\delta-1) + 2^{J-j}l.
		\end{aligned}
		\right.
	\end{equation*}
	where the inequalities have to be understood component-wise. These conditions are equivalent to 
	\begin{equation*}
		2^{j-J}\left( q_0 - (2^{J-j} -1) \right) \leq l \leq 2^{j-J} \left( q_1 + (1-2^{J}) 2 (\delta-1)\right).
	\end{equation*}
	The volume of such a set is
	\begin{equation*}
	\begin{split}
		\#\supp c_{j,e} & = \left( 2^{j-J} q_1 - 2^{j-J} q_0 + (1-2^{-J}) 2(\delta-1) + (1-2^{-J}) + 1 \right)^d \\
		& = \left( 2^{j-J} (q_1-q_0+1) + (1-2^{-J}) 2 \delta \right)^d \\ 
		& \leq 2^d \left( 2^{d(j-J)} \kappa^d + (2 \delta)^d \right)
	\end{split}
	\end{equation*}

	Since filters $g$ and $h$ have the same length, the details and average coefficients of the wavelet transforms have the same size in each sub-band - albeit not the same support. 
	Therefore we only analyze the detail coefficients.
	The computation of one $c_{j,e}$ is obtained from $d$ convolutions of average coefficients at scale $j+1$ with wavelet filters of size $2\delta$. Hence the number of operations to compute one $c_{j,e}$ is $d 2 \delta \# \supp c_{j,e}$ and the total number of operations needed to compute $c$ is
	\begin{equation}
		\begin{split}
			& \sum_{j=0}^{J-1} (2^d-1) 2 d \delta \# \supp c_{j,e} \\
			& \leq 2 d \delta (2^d-1) 2^d \sum_{j=0}^{J-1} \left( 2^{d(j-J)} \kappa^d + (2 \delta)^d \right) \\
			& \leq  2 d \delta (2^d-1) 2^d \left(  \kappa^d \frac{1}{2^d - 1} +  (2 \delta)^d J \right) \\
		\end{split}
	\end{equation}
	which ends the proof.
\end{proof}

We are now ready to bound the complexity of step 1.
\begin{lemma} \label{lem:comp_hatV}
  The computation of the matrix $\widehat V_k$ requires at most 
		\begin{equation*}
			c(d) (2 \delta)^{d+1} N \log_2 N 
		\end{equation*}
		operations. The constant $c(d) = 2^{d} (2^d-1)$.
\end{lemma}
\begin{proof}
  The $i$-th column of $V$ is supported on $\{\omega^{-1}(i)\}$. Using Lemma \ref{lem:wavelet_transform} together with $J = \frac{1}{d} \log_2 N$ give the result.
\end{proof}

The cost of step 2 is controlled by the following Lemma.
\begin{lemma} \label{lem:cost_B}
  The computation of $B_k$ from $\widehat V_k$ requires at most
  \begin{equation*}
  c(d) (2 \delta)^{d+1} N \log_2 N
  \end{equation*}
  operations. The constant $c(d) = 2^{d+1} (2^d-1)$.
\end{lemma}
\begin{proof}
	To bound the number of operations needed to compute $B$ from $\widehat V$, we need to know the support of each row of $\widehat V$.
	Let $\lambda = (j,e,l)$. From Lemma \ref{lem:wavelet_transform}, the necessary condition for $\widehat V [\lambda,i]$ to be non zero is
	\begin{equation*}
		l \in \left[ 2^{j-J} \omega^{-1}(i) - (1-2^{-J}), 2^{j-J} \omega^{-1}(i) + (1-2^{-J})2(\delta-1) \right]
	\end{equation*}
	Equivalently the condition becomes
	\begin{equation*}
		\omega^{-1}(i) \in \left[ 2^{J-j} l - (2^{J-j}-1) 2 (\delta-1), 2^{J-j} l + (2^{J-j} - 1) \right].
	\end{equation*}
	Therefore $\supp \widehat V[\lambda, \cdot] = \left[ 2^{J-j} l - (2^{J-j}-1) 2 (\delta-1), 2^{J-j} l + (2^{J-j} - 1) \right]$.
	Using again Lemma \ref{lem:wavelet_transform} we conclude that the number of operations required to compute $B$ from $\widehat V$ is bounded above by
	\begin{equation*}
		\begin{split}
			& \sum_{j=0}^{J-1} (2^d-1) 2^{jd} 2 d \delta 2^d \left(  ( (2^{J-j} - 1) (2\delta-1) + 1)^d +  (2^d-1) (2 \delta)^d J \right) \\
			& \leq (2^d-1) 2 d \delta 2^d \sum_{j=0}^{J-1} 2^{jd} \left( 2^{d(J-j)} (2\delta)^d  + (2^d-1) (2 \delta)^d J \right) \\
			& \leq (2^d-1) 2 d \delta 2^d \left( (2\delta)^d J 2^{dJ} + (2 \delta)^d J 2^{dJ} \right) \\
			& \leq 2 2^d (2^d-1) (2 \delta)^{d+1} N \log_2 N \\
		\end{split}
	\end{equation*}
	where we used $J = \frac{1}{d} \log_2 N$ to conclude the proof.
  \end{proof}
  
Overall, the computation of one matrix $B_k$ can be achieved in $O\left(\delta^{d+1} N \log_2 N \right)$ operations.
As seen in this two-step algorithm, the crucial ingredient of its efficiency is the sparse cascade algorithm. 
Its implementation is very similar to the standard algorithm but loops on the non-zero elements of the input signal in $\mathcal{E}$. 
Furthermore, the indexes of each non-zero coefficient at the next scale have to be tracked and reused at the next filtering stage.

\subsection{Products \texorpdfstring{$\widetilde{A}_k B_k$}{AB}}

Two questions are important while investigating the product $\widetilde{A}_k B_k$. First, how many operations are needed for its computation and how many non-zero coefficients are in the product matrix. 
These questions are answered in the following Lemma.

\begin{lemma} \label{cor:cost_AB} \sloppy
	The number of operations required to compute $\widetilde{A}_k B_k$ is bounded above by $c(d) \delta^d N \log_2 N \epsilon_k^{-d/\alpha}$.
	Furthermore, the number of non-zero coefficients is bounded above by the same quantity.
	The constant $c(d) = \frac{2}{d} (2^d - 1) 2^d$. 
\end{lemma}

The product $\widetilde{A}_k B_k$ could be implemented by building the matrix $\widetilde A_k$ from the set of circulant vectors and then rely on a standard implementation of sparse matrix-matrix products. 
However, the building step could be memory-intensive and long time-wise.
This step is avoided by taking advantage of the sub-band circulant structure of $\widetilde{A}_k$: locations of non-zero coefficients can be predicted from the elements of the circulant vectors.

Before proving Lemma \ref{cor:cost_AB}, we present a useful Lemma.

\begin{lemma}[{\cite{gustavson1978two,yuster2005fast}}] \label{lem:complexity_product_sparse}
	Let $A$ and $B$ two $N \times N$ matrices.
	Let $\alpha_i$ be the number of non-zero entries in the $i$-th column of $A$ and $\beta_i$ be the number of non-zeros entries in the $i$-th row of B.
	Then, the number of operations required to perform $AB$ is only $\sum_{i=1}^N \alpha_i \beta_i$.
	Furthermore, the number of non-zero coefficients in $AB$ is bounded above by the same quantity.
\end{lemma}

We are now ready to prove Lemma \ref{cor:cost_AB}.
\begin{proof}[Proof of Lemma \ref{cor:cost_AB}]
	From Theorem \ref{thm:cohen} we know that the number of non-zero coefficients in the $\lambda$-th column of A is bounded by 
	\begin{equation*}
		\alpha_\lambda = L.
	\end{equation*}
	Similarly, Lemma \ref{lem:num_coeff_row_B} shows that the number of coefficients in the $\lambda$-th row of B is bounded by 
	\begin{equation*}
		\beta_\lambda = 2^d \delta^d \left( (2^d-1) j + 2^{d(J-j)}\right).
	\end{equation*}
	Therefore the total number of operations is bounded above by:
	\begin{equation*}
		\begin{split}
		\sum_{\lambda \in \Lambda} \alpha_\lambda \beta_\lambda & \leq \sum_{j=0}^{J-1} (2^d - 1) 2^{jd} L 2^d \delta^d \left( (2^d-1) j + 2^{d(J-j)}\right) \\
		& = (2^d - 1) 2^d \delta^d L \left( (2^d-1) \sum_{j=0}^{J-1} j 2^{jd} +  \sum_{j=0}^{J-1} 2^{dJ} \right) \\
		& \leq 2 (2^d - 1) 2^d \delta^d L J 2^{dJ} \\
		& \leq \frac{2}{d} (2^d - 1) 2^d \delta^d L N \log_2 N.
		\end{split} 
	\end{equation*}	
	Furthermore, $L$ is of the order $\epsilon_k^{-d/\alpha}$ to reach an accuracy $\epsilon_k$ on the approximation of $A_k$. Therefore the total number of operations is bounded above by $c(d) \delta^d N \log_2 N \epsilon_k^{-d/\alpha}$.
\end{proof}

\subsection{Accumulation of \texorpdfstring{$\widetilde{A}_k B_k$}{AB}}

The last step of the algorithm consists in summing all the products $C_k = \widetilde{A}_k B_k$ into $\widetilde \Theta$.
The following results show that the number of coefficients in $\widetilde \Theta$ does not explode while adding the $C_k$ and that the returned $\widetilde \Theta_k$ reaches the prescribed accuracy.

\begin{lemma} \label{lem:support_C}
	There exists $\mathcal{S} \subset \Lambda^2$ s.t. $\supp \widetilde{A}_k B_k \subset \mathcal{S}$ for all $1 \leq k \leq m$.
	Moreover, $\# \mathcal{S} = O\left(c(d) \delta^d N \log_2 N (\min_k  \epsilon_k)^{-d/\alpha} \right)$.
\end{lemma}
\begin{proof}
	Let $k_0$ be the index such that $\epsilon_{k_0} = \min_k \epsilon_k$.
	Since $A_{k_0}$ belongs to the smoothness class $\mathcal{A}_{\alpha}$, there exists a $\mathcal{F} \subset \Lambda^2$ s.t. $\supp \widetilde{A}_{k_0} \subseteq \mathcal{F}$. 
	Furthermore all $A_k$ belong to the same $\mathcal{A}_{\alpha}$ we have that $\supp \widetilde A_k \subseteq \supp \widetilde A_{k_0}$.
	
	On the other hand, since all $V_k$ are diagonal matrices $\supp B_k \subseteq \mathcal{I}$ for all $k$. 
	
	Define $I$ (resp. $F$) a $N \times N$ matrix with ones on $\mathcal{I}$ (resp. $\mathcal{F}$) and zeros outside. Let $\mathcal{S} = \supp FI$.
	Then obviously $\supp \widetilde A_k B_k \subseteq \mathcal{S}$. Moreover, from Lemma \ref{cor:cost_AB}, the number of non-zero coefficients in $FI$ is bounded above by $c(d) \delta^d N \log_2 N \epsilon_{k_0}^{-d/\alpha}$ which concludes the proof.
\end{proof}

\subsection{Gathering everything}

We will now gather all the ingredients to prove Theorem \ref{thm:overall_complexity}. The first part of the proof certifies the accuracy of the approximation.

\begin{lemma} \label{lem:accuracy_theta}
	The matrix $\widetilde{\Theta}$ returned by Algorithm \ref{alg:decomp} satisfies $\| \Theta - \widetilde{\Theta} \|_{2 \to 2} \leq \eta$
\end{lemma}

\begin{proof}
	Using triangular inequality and since $\| \cdot\|_{2\to2}$ is multiplicative we get
	\begin{equation*}
		\begin{split}
		\| \Theta - \widetilde{\Theta} \|_{2 \to 2} & \leq \sum_{k=1}^m \| A_k B_k -\widetilde{A}_k B_k \|_{2\to 2} \\
		& \leq \sum_{k=1}^m \| B_k\|_{2\to 2} \| A_k -\widetilde{A}_k \|_{2\to 2} \leq \sum_{k=1}^m \| V_k\|_{2\to 2} \epsilon_k \\
		& = \sum_{k=1}^m \| v_k \|_{\infty} \epsilon_k = \eta.\\
		\end{split}
	\end{equation*}
	We used the fact that $\| B_k \|_{2 \to 2} = \| V_k\|_{2 \to 2} = \| v_k \|_{\infty}$ since $\Psi$ is an orthogonal wavelet basis and $V_k$ is a diagonal matrix. The last equality is derived from the definition of $\epsilon_k = \frac{\eta}{m \| v_k\|_{\infty}}$.
\end{proof}

The second part deals with the overall complexity.
\begin{proof}[Proof of Theorem \ref{thm:overall_complexity}]
	The overall complexity of computing $\widetilde \Theta$ is bounded above by a quantity proportional to
	\begin{equation*}
		\begin{split}
		& c(d) \left( \sum_{k=1}^m \underbrace{N \log^2_2 N}_{\text{Lemma }\ref{lem:struct_A}} + N \log_2 N + \underbrace{\delta^{d+1}  N \log_2^2 N}_{\text{Lemma } \ref{lem:cost_B}} + 2 \underbrace{\delta^d N \log_2 N \epsilon_k^{-d/\alpha}}_{\text{Lemma } \ref{cor:cost_AB}} \right)\\
		& \leq c(d) \left(3 \delta^{d+1} m N \log^2_2 N + 2 N \log_2 N \sum_{k=1}^m \left( \frac{\eta}{m \| v_k\|_{\infty}} \right)^{-d/\alpha} \right)\\
		& \leq c(d) \left(\delta^{d+1} m N \log^2_2 N + N \log_2 N \eta^{-d/\alpha} m^{d/\alpha} \left( \sum_{k=1}^m \| v_k\|_{\infty}^{d/\alpha} \right) \right)\\
		\end{split}
	\end{equation*}
	The constant $c(d)$ depends only on $d$. The term $\left( \sum_{k=1}^m \| v_k\|_{\infty}^{d/\alpha} \right)$ depends on the considered operator $H$. 
	Finally, $\widetilde \Theta$ is computed in $O\left(\delta^{d+1} m^{\max(1,d/\alpha)} N \log^2_2 N \eta^{-d/\alpha} \right)$.

	The bound on the number of coefficients in $\widetilde \Theta$ is obtained using Lemma \ref{lem:support_C}. Since there exists a support $\mathcal{S}$ such that all $\supp C_k \subseteq \mathcal{S}$ we get that $\supp \widetilde \Theta \subseteq \mathcal{S}$.
\end{proof}

\begin{remark}
	The complexity analysis does not take into account that $v_k$ may have some smoothness properties. In this case, the associated $B_k$ might belong to a class $\mathcal{A}_{\alpha'}$ and could therefore be thresholded to further speed up the algorithm.
\end{remark}

\bibliographystyle{abbrv}
\bibliography{biblio}

\appendix

\section{Deblurring algorithms} \label{sec:deblurring_algos}


We briefly describe the deblurring algorithms used in the numerical experiments.

\subsection{FISTA}

Let $\Psi:\R^N\to \R^N$ denote an orthogonal wavelet transform and $H:\R^N\to \R^N$ denote a linear operator. 
We then have
\begin{equation} \label{eq:optim_deblurring_generic}
    \argmin_{f \in \R^N} \frac{1}{2} \left\| H f - f_0 \right\|_2^2 + \left\| \Psi^* f \right\|_{1,w} = \Psi \left(\argmin_{z \in \R^N} \frac{1}{2} \left\| H \Psi z - f_0 \right\|_2^2 + \left\| z \right\|_{1,w}\right).
\end{equation}
The Forward-Backward algorithm applied to the right hand-side of \eqref{eq:optim_deblurring_generic} starts with an initial guess $z^{(0)} = y^{(1)} \in \R^N$ and iterates for $i \geq 1$ as
\begin{equation*}
    \left\{
    \begin{aligned}
        z^{(i)} &=  \textrm{Prox}_{\| \cdot\|_{1,w}} \left( y^{(i)}  - \tau \Psi^* H^*(H \Psi y^{(i)} - f_0)\right) \\
        y^{(i+1)} &= z^{(i)} + \frac{i-1}{i+2} \left( z^{(i)} - z^{(i-1)}\right)
    \end{aligned}
    \right.
\end{equation*}    
with a step-size $\tau \leq \| H \|_{2 \to 2}^{-2}$. This algorithm can be applied verbatim with the exact spatial operator (coined as FISTA-Spa). In the same fashion, we design another algorithm, named FISTA-PC, using the product-convolution expansion.

One can also use FISTA to solve \eqref{eq:optim_deblurring_wavelet}, the approximation of \eqref{eq:optim_deblurring} in the wavelet domain. It starts by setting $z^{(0)} = y^{(1)} \in \R^N$ and iterates for $i \geq 1$ as
\begin{equation*}
    \left\{
    \begin{aligned}
        z^{(i)} &=  \textrm{Prox}_{\| \cdot\|_{1,w}}^\Sigma \left( y^{(i)}  - \tau \Sigma \Theta_L^*(\Theta_L y^{(i)} - z_0) \right) \\
        y^{(i+1)} &= z^{(i)} + \frac{i-1}{i+2} \left( z^{(i)} - z^{(i-1)}\right)
    \end{aligned}
    \right.
\end{equation*}    
where $\Sigma$ is a diagonal preconditioner equal to the identity or to the Jacobi preconditioner $\mathrm{diag}(\Theta_L^* \Theta_L)$. Note that in this alternative, one iteration costs only 2 matrix-vector products with the sparse matrix $\Theta_L$. We will refer to this algorithm as FISTA-W when $\Sigma=\mathrm{I}_N$ and FISTA-WP when $\Sigma=\mathrm{diag}(\Theta_L^* \Theta_L)$ (see \cite{escande2018accelerating} for more details).

\subsection{ADMM for product-convolution}
\label{sec:admm}

The resolution of \eqref{eq:optim_deblurring} using an ADMM or a Douglas-Rachford is expensive due to the resolution of non trivial linear systems at each iteration. We can adapt ideas proposed in \cite{o2017total} and exploit the product-convolution structure to accelerate the resolution of the linear systems and improve its efficiency. We describe this idea below.

A product-convolution operator $H$ can be decomposed as:
\begin{equation*}
 Hf=\sum_{k=1}^m u_k\star v_k \odot f =\mathcal{U}\mathcal{V} f,
\end{equation*}
with $\mathcal{V} = \begin{pmatrix} V_1 \ldots V_m \end{pmatrix}^T \in \R^{m N \times N}$ and $\mathcal{U} = \begin{pmatrix} U_1 \ldots U_m \end{pmatrix} \in \R^{N \times mN}$. This decomposition allows to split the problem \eqref{eq:optim_deblurring_generic} as
\begin{equation*}
    \min_{\substack{f \in \R^N, g \in \R^{mN}, z \in \R^N \\ z = \Psi^* f, g = \mathcal{V}f}} \frac{1}{2} \| \mathcal{U} g - f_0\|_2^2 + \| z \|_{1,w}
\end{equation*}

The augmented Lagrangian associated to this problem reads
\begin{equation*}
    \mathcal{L}(f,g,z,\xi,\zeta) = \frac{1}{2} \| \mathcal{U} g - f_0\|_2^2 + \| z \|_{1,w} + \left\langle \xi, g - \mathcal{V} f \right \rangle + \left\langle \zeta, z - \Psi^*f \right \rangle 
    + \frac{\beta_1}{2} \| g - \mathcal{V}f \|_{2,\gamma}^2 + \frac{\beta_2}{2} \| z - \Psi^*f \|_2^2.
\end{equation*}
The norm $\| \cdot \|_{2,\gamma}$ is just a weighted $l^2$-norm for the augmented term of $g = \mathcal{U}f$. Weighting the influence of each $U_k$ allows accelerating the convergence of the ADMM, i.e. it acts as a preconditioner. 
This norm is defined by $\| \cdot \|_{2,\gamma}^2 = \langle D_{\gamma} \cdot, \cdot \rangle_{\R^{mN}}$, with 
\begin{equation*}
    D_\gamma = \begin{pmatrix}
        \gamma_1 \Id_N & 0 & 0 \\
        0 & \ddots & 0 \\
        0 & 0 & \gamma_m \Id_N
    \end{pmatrix} \in \R^{mN \times mN}
\end{equation*}
a block diagonal matrix where $\gamma_k > 0$ and $\sum_{k=1}^m \gamma_k = 1$. The choice $\gamma_k = \| U_k \|_{2 \to 2}$ led  to good practical behaviors.

The ADMM amounts to iterating:
\begin{itemize}
    \item $g^{(i+1)} = \argmin_{g \in \R^{mN}} \mathcal{L}\left(f^{(i)},g,z^{(i)},\xi^{(i)},\zeta^{(i)} \right)$. 

    The solution of this sub-problem is given by 
    \begin{equation*}
        g^{(i+1)} = (\mathcal{U}^* \mathcal{U} + \beta_1 D_\gamma)^{-1} \left(\mathcal{U}^* f_0 - \xi^{(i)}  + \beta_1 D_\gamma \mathcal{V} f^{(i)} \right).
    \end{equation*}
    The matrix $\mathcal{U}^* \mathcal{U}$ contains $m \times m$ blocks of size $N \times N$ with the $(k,l)$-th block populated with $U_k^* U_l$. The direct inversion of this matrix will be inefficient even in the Fourier domain. However, using the Woodbury matrix identity \cite{hager1989updating}, the solution can be expressed as
    \begin{equation*}
        g^{(i+1)} = \left( D_\gamma^{-1} - D_\gamma^{-1} \mathcal{U^*} (\Id_{N} + \mathcal{U} D_{\gamma}^{-1} \mathcal{U}^*)^{-1} \mathcal{U} D_{\gamma}^{-1} \right) \left(\mathcal{U}^* f_0 - \xi^{(i)}  + \beta_1 D_\gamma \mathcal{V} x^{(i)} \right).
    \end{equation*}
    Matrix $(\Id_{N} + \mathcal{U} D_{\gamma}^{-1} \mathcal{U}^*) = \Id_{N} + \sum_{k=1}^m \gamma_k U_k U_k^* \in \R^{N \times N}$ is diagonal in the Fourier domain and can be efficiently inverted. 
    In this step, the computation of $3 m$ FFTs and $4m + 1$ products with diagonal matrices are involved.
    
    \item $z^{(i+1)} = \argmin_{z \in \R^{N}} \mathcal{L}\left(f^{(i)},g^{(i+1)},z,\xi^{(i)},\zeta^{(i)} \right) = \textrm{Prox}_{\frac{1}{\beta_2} \| \cdot \|_{1,w}} \left( \Psi^* f^{(i)} - \frac{1}{\beta_2} \zeta^{(i)} \right)$. 
   
   This step costs one wavelet tranform plus a thresholding operation.
    
    \item $f^{(i+1)} = \argmin_{f \in \R^{N}} \mathcal{L}\left(f,g^{(i+1)},z^{(i+1)},\xi^{(i)},\zeta^{(i)} \right)$.

    The solution of this sub-problem reads 
    \begin{equation*}
     f^{(i+1)} = \left(\beta_1 \mathcal{V}^* D_\gamma \mathcal{V} + \beta_2 \Id_{N} \right)^{-1} \left( \mathcal{V}^* (\xi^{(i)} + \beta_1 D_\gamma g^{(i+1)}) + \Psi (\beta_2 z^{(i+1)}  + \zeta^{(i)}) \right).
    \end{equation*}
    Matrix $\beta_1 \mathcal{V}^* D_\gamma \mathcal{V} + \beta_2 \Id_{N} = \beta_1 \sum_{i=1}^m \gamma_i V_i V_i + \beta_2 \Id_{N}$ is $N \times N$ diagonal matrix and can be efficiently inverted.

    In this step, one wavelet transform and $m + 1$ products with diagonal matrices are involved.
    \item Update Lagrange multipliers
    \begin{equation*}
        \begin{split}
            \xi^{(i+1)} & = \xi^{(i)} + \beta_1 D_{\gamma} (g^{(i+1)} - \mathcal{V} f^{(i+1)}) \\
            \zeta^{(i+1)} & = \zeta^{(i)} +  \beta_2 (z^{(i+1)} - \Psi^* f^{(i+1)}).
        \end{split}
    \end{equation*}
    This step requires $m$ products with diagonal matrices and one wavelet transform.
\end{itemize}

To summarize, one iteration of the ADMM costs $3 m$ FFTs, $3$ wavelet transforms and $6 m + 2$ products with diagonal matrices. This is substantially more than one interation of FISTA with the product-convolution expansion i.e. $2m$ FFTs, $2$ wavelet transforms and $m$ products with diagonal matrices.

\section*{Acknowledgments}
P. Weiss is supported by the ANR JCJC Optimization on Measures Spaces ANR-17-CE23-0013-01 and the ANR-3IA Artificial and Natural Intelligence Toulouse Institute. Both authors thank the GDR ISIS its support.

\end{document}